\documentclass[11pt]{article}

\usepackage{libertine}

\usepackage{amsmath,amssymb,amsfonts,amsthm}
\usepackage[shortlabels]{enumitem}
\usepackage[margin=1in]{geometry}

\usepackage{mathtools}
\usepackage{mathrsfs}

\usepackage[colorlinks = true,linkcolor = black!20!purple, urlcolor = blue, citecolor = black!30!red]{hyperref}
\usepackage[capitalise,nameinlink]{cleveref}

\usepackage[textsize=footnotesize]{todonotes}

\usepackage{caption}

\usepackage{thm-restate}

\usepackage{comment}

\usepackage{textpos}
\newtheorem{theorem}{Theorem}
\newtheorem{definition}[theorem]{Definition}
\newtheorem{lemma}[theorem]{Lemma}
\newtheorem*{lemma*}{Main Lemma}
\newtheorem{corollary}[theorem]{Corollary}

\newtheorem{remark}[theorem]{Remark}

\newtheorem{claim}[theorem]{Claim}

\newcommand{\Aa}{\mathcal{A}}
\newcommand{\N}{\mathbb{N}}

\newcommand{\C}{\mathscr{C}}
\newcommand{\Cc}{\C}
\newcommand{\D}{\mathscr{D}}
\newcommand{\Dd}{\D}
\newcommand{\E}{\mathscr{E}}
\newcommand{\Ee}{\E}
\newcommand{\F}{\mathcal{F}}
\newcommand{\PP}{\mathcal{P}}
\newcommand{\B}{\mathcal{B}}

\newcommand{\Tf}{\mathsf{T}}

\newcommand{\T}{\mathsf{T}}

\newcommand{\ST}{\mathsf{S}}
\newcommand{\ET}{\mathsf{E}}

\newcommand{\FO}{\mathsf{FO}}

\newcommand{\leader}{\mathsf{leader}}
\newcommand{\branch}{\mathsf{branch}}
\newcommand{\adj}{\mathsf{adj}}

\newcommand{\double}{\mathsf{Double}}

\newcommand{\Oh}{\mathcal{O}}
\newcommand{\eps}{\varepsilon}

\renewcommand{\phi}{\varphi}

\newcommand{\grad}{\nabla}

\renewcommand{\leq}{\leqslant}
\renewcommand{\geq}{\geqslant}
\renewcommand{\le}{\leqslant}
\renewcommand{\ge}{\geqslant}

\newcommand{\aff}[1]{\textcolor{black!60}{\small{#1}}}

\title{Efficient reversal of transductions of sparse graph classes\footnote{\funding}}
\date{}

\author{
	Jan Dreier \\
	\aff{TU Wien} \\
	\aff{dreier@ac.tuwien.ac.at}
	\and
	Jakub Gajarsk\'y \\
	\aff{Masaryk University and}\\
	\aff{University of Warsaw} \\
	\aff{gajarsky@mimuw.edu.pl}
	\and
	Michał Pilipczuk \\
	\aff{University of Warsaw} \\
	\aff{michal.pilipczuk@mimuw.edu.pl}
}

\begin{document}

\newcommand{\funding}{JG is supported by the Polish National Science Centre
SONATA-18 grant number 2022/47/D/ST6/03421.
MP is supported by the project BOBR that has received funding from the European Research Council (ERC) under the European Union’s Horizon 2020 research and innovation programme, grant agreement No. 948057.}

\maketitle

\begin{abstract}
	(First-order) transductions are a basic notion capturing graph modifications that can be described in first-order logic. In this work, we propose an efficient algorithmic method to approximately reverse the application of a transduction, assuming the source graph is sparse. Precisely, for any graph class $\Cc$ that has structurally bounded expansion (i.e., can be transduced from a class of bounded expansion), we give an $\Oh(n^4)$-time algorithm that given a graph $G\in \Cc$, computes a vertex-colored graph $H$ such that $G$ can be recovered from $H$ using a first-order interpretation and $H$ belongs to a graph class $\Dd$ of bounded expansion. This answers an open problem raised by Gajarsk\'y et al.~[ACM TOCL,~'20]. In fact, for our procedure to work we only need to assume that $\Cc$ is monadically stable (i.e., does not transduce the class of all half-graphs) and has inherently linear neighborhood complexity (i.e., the neighborhood complexity is linear in all graph classes transducible from~$\Cc$). This renders the conclusion that the graph classes satisfying these two properties coincide with classes of structurally bounded~expansion.
\end{abstract}

\begin{textblock}{20}(-1.75, 5.3)
	\includegraphics[width=40px]{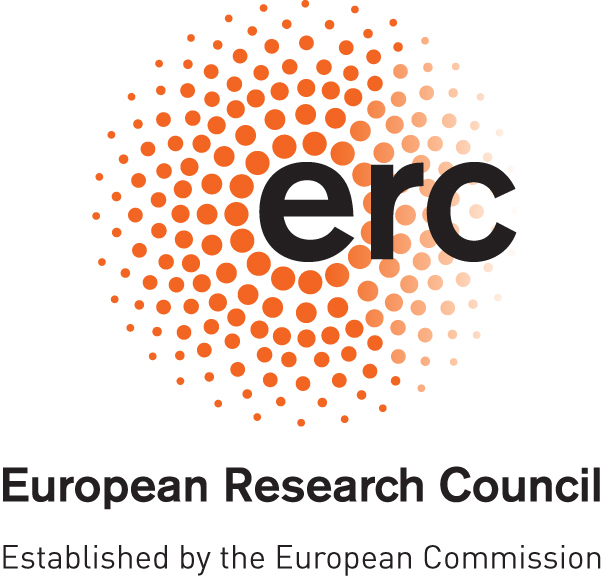}
\end{textblock}
\begin{textblock}{20}(-1.75, 6.3)
	\includegraphics[width=40px]{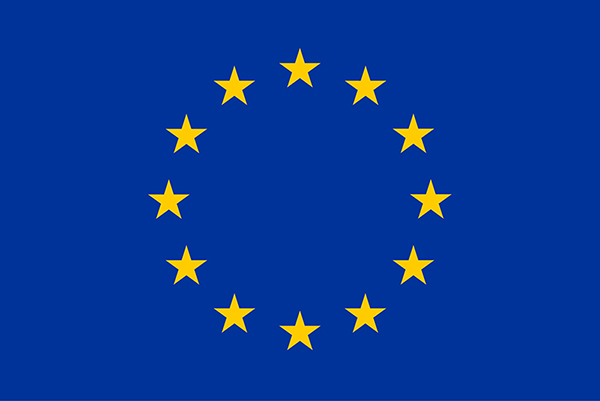}
\end{textblock}

\section{Introduction}

Consider the operation of {\em{squaring}} a graph: given a graph $H$, compute the graph $G\coloneqq H^2$ on the same vertex set such that two vertices are adjacent in $G$ if and only if they are at distance at most $2$ in $H$. It should be expected that if the source graph $H$ belongs to some class of well-behaved graphs, for instance is planar, or tree-like, or in some other way sparse, then the output graph $G$ retains much of the useful structure of $H$. Exposing this structure would be of course easy if $G$ was provided along with $H$ on the input. But given $G$ alone, how would one compute the source graph $H$ efficiently? For the concrete case of squaring a graph, this problem
(given \(G\), find an \(H\) with \(G=H^2\)) is actually $\mathsf{NP}$-hard in general~\cite{squares_hardness}.

The operation of squaring a graph, discussed above, is a basic example of a more general framework of graph modification operations that can be described in logic: first-order interpretations and transductions. A (simple\footnote{The standard notion of a first-order interpretation discussed in the literature on logic allows the vertices of the output graph to be $k$-tuples of vertices of the input graph, for some fixed $k\in \N$. In this paper we restrict attention to {\em{simple}} or {\em{one-dimensional}} interpretations, that is, those with $k=1$.}) first-order interpretation $I$ is described by a pair $\psi(x,y),\nu(x)$ of formulas of first-order logic on graphs (called further $\FO$), where $\psi$ has two free variables and $\nu$ has one. Applying $I$ to a graph $H$ yields a new graph $G=I(H)$ obtained from $H$ by (i) replacing the adjacency relation by making any two vertices $u,v$ adjacent in $G$ if and only if $\psi(u,v)\vee \psi(v,u)$ holds in $H$, and (ii) restricting the vertex set to those vertices $u$ for which $\nu(u)$ holds in $H$. Thus, to model the operation of squaring a graph as a first-order interpretation we would take $\psi(x,y)=\adj(x,y)\vee \exists_z (\adj(x,z)\wedge \adj(y,z))$ and $\nu(x)=\mathtt{true}$. A {\em{first-order transduction}} is a more general notion that additionally allows (i) copying the vertex set a fixed number of times, and (ii) marking the vertices with a fixed set of labels (called {\em{colors}}) that can be accessed in the formulas $\psi$ and~$\nu$. See \Cref{sec:prelims} for a formal definition. In this work we will consider only first-order interpretations and transductions, hence from now on we drop the prefix {\em{first-order}} for brevity.

For graph classes $\Cc$ and $\Dd$, we say that $\Cc$ is {\em{transducible}} from $\Dd$ if there exists a transduction $\Tf$ such that $\Cc\subseteq \Tf(\Dd)$, where $\Tf(\Dd)\coloneqq \bigcup_{H\in \Dd} \Tf(H)$ comprises of all the possible outputs of $\Tf$ applied to a graph from~$\Dd$. With this notion recalled, we can make the main questions motivating our work more concrete:
\begin{enumerate}
\item Assuming $\Dd$ is a well-behaved class of graphs and $\Cc$ is transducible from $\Dd$, what can we say about the structure of graphs in $\Cc$? 
\item Can we efficiently revert the transduction, that is, given a graph $G\in \Cc$ compute a colored graph $H\in \Dd$ in which $G$ can be interpreted?
\end{enumerate}

The first, graph-theoretic question has been pivotal for the development of the theory of {\em{structurally sparse graph classes}}, that is, those transducible from classes of sparse graph. Of particular importance to our motivation is the work of Gajarsk\'y, Kreutzer, Ne\v{s}et\v{r}il, Ossona de Mendez, Pilipczuk, Siebertz, and Toru\'nczyk~\cite{gajarsky2020sbe}, who investigated classes of structurally bounded expansion and characterized them through the existence of low shrubdepth covers. Here, a class $\Cc$ has {\em{structurally bounded expansion}} if it is transducible from some class $\Dd$ of bounded expansion, whereas $\Dd$~has {\em{bounded expansion}} if for every $d\in \N$ there is a finite bound $c(d)\in \N$ on the average degree of graphs that can be obtained from graph from $\Dd$ by first deleting some vertices and then contracting mutually disjoint subgraphs of radius at most~$d$. Bounded expansion is the central notion of uniform sparsity in graphs considered in the theory of Sparsity of Ne\v{s}et\v{r}il and Ossona de Mendez, and it covers all classes that forbid a fixed topological minor, numerous classes of sparse graphs coming from geometry, and many other examples; see~\cite{sparsity,sparsityNotes} for an introduction to the theory of Sparsity. 
Further combinatorial descriptions of classes of structurally bounded expansion were given by Dreier~\cite{lacon_decompositions} (through {\em{lacon-}} and {\em{parity decompositions}}) and by Dreier, Gajarsk\'y, Kiefer, Pilipczuk, and Toru\'nczyk~\cite{bushes_quasibushes} (through bounded-depth {\em{bushes}}). Structural characterizations are also known for classes of structurally bounded degree~\cite{gajarsky2020bd_interp}, structurally bounded pathwidth~\cite{nesetril2021linrw_stable}, structurally bounded treewidth~\cite{DBLP:conf/soda/NesetrilMPRS21}, structurally bounded twin-width~\cite{GajarskyPT22}, and structurally nowhere dense~classes~\cite{bushes_quasibushes}.

The second, algorithmic question remains so far much less understood. Suitable polynomial-time algorithms have been proposed for classes of structurally bounded degree~\cite{gajarsky2020bd_interp}, their ``second-level'' generalizations --- classes of structurally tree-rank at most $2$~\cite{RoseLemma}, classes of structurally bounded pathwidth~\cite{nesetril2021linrw_stable}, and classes of structurally bounded treewidth~\cite{DBLP:conf/soda/NesetrilMPRS21}. In all these cases, the discussed {\em{sparsification procedures}} adhere to the following template: Consider a graph class $\Cc$ that is structurally $\PP$, that is, $\Cc$ is of the form $\Cc \subseteq \T(\D)$, where $\D$ is a graph class enjoying property $\PP$, with ${\cal P}\in \{\textrm{bounded degree, tree-rank}\leq 2\textrm{, bounded pathwidth, bounded treewidth}\}$. Then the sparsification procedure computes a colored graph $H$ belonging to some fixed class $\Dd'$ enjoying $\cal P$ such that $G=I(H)$ for some fixed interpretation $I$. Note that 
this means that the output graph $H$ will in general \emph{not} be from the original `source' graph class $\Dd$.

Namely, the sparsification procedure may (and usually will) come up with another class $\Dd'$ enjoying $\cal P$, and another interpretation $I$ different from the one used in $\Tf$, such that every graph in $\Cc$ can be interpreted from a colored graph belonging to $\Dd'$ using $I$. As shown in~\cite{gajarsky2020bd_interp}, this relaxation of the concept of reversing a transduction is necessary to obtain tractability, because the problem of reversing a fixed interpretation faithfully is $\mathsf{NP}$-complete already when the source class comprises of graphs of maximum degree at most~$3$.

For the case of classes of structurally bounded expansion, none of the existing combinatorial characterizations~\cite{lacon_decompositions,bushes_quasibushes,gajarsky2020sbe} provides an efficient sparsification procedure adhering to the template above. The question about the existence of such a procedure was asked by Gajarsk\'y et al. in~\cite{gajarsky2020sbe}, and has been circulating in the community since the publication of this work. We remark that recently, Braunfeld, Ne\v{s}et\v{r}il, Ossona de Mendez, and Siebertz~\cite{horizons} gave\footnote{In~\cite{horizons}, Braunfeld et al. state only a combinatorial result. However, all the tools they use are algorithmic and hence their reasoning can be turned into a polynomial-time sparsification procedure.} a sparsification procedure that applies in an even larger generality of {\em{monadically stable}} classes. However the output of their procedure is only guaranteed to belong to a class that is {\em{almost nowhere dense}}, even if the input class actually has structurally bounded expansion. In essence, this means that in the output graph, all the relevant structural parameters, such as the weak coloring numbers, are only guaranteed to be bounded by subpolynomial functions of the graph size, instead of by constants, as would be the case in a bounded expansion class. In particular, this severely limits the usability of such sparsification in the algorithmic context,
and does not yield the algorithmic implications discussed further below.

\paragraph*{Our contribution.} We answer the question of Gajarsk\'y et al.~\cite{gajarsky2020sbe} by proving the following result.

\begin{restatable}{theorem}{thmmain}\label{thm:main}
    Let $\C$ be a monadically stable graph class with inherently linear neighborhood complexity.
    Then there exists a graph class $\D$ of bounded expansion transducible from $\Cc$, an algorithm $\Aa$, and a first-order interpretation $I$ such that the following holds: For any input graph $G \in \C$ on $n$ vertices, the algorithm $\Aa$ computes in time $\Oh(n^4)$ a colored graph $H \in \D$ such that $G = I(H)$.
\end{restatable}

Let us explain the so-far undefined terms used in \cref{thm:main}. First, a graph class $\Cc$ is {\em{monadically stable}} if $\Cc$ does not transduce the class of all half-graphs. (A {\em{half-graph}} of order $n$ is a bipartite graph with sides $\{u_1,\ldots,u_n\}$ and $\{v_1,\ldots,v_n\}$ where $u_i$ and $v_j$ are adjacent iff $i\leq j$.) This notion appears to play a central role in the theory of transducibility, see~\cite[Section 4.1.1]{michal_survey} for a comprehensive introduction. In particular, every class of structurally bounded expansion is monadically stable. Second, we say that a graph class $\Cc$ has {\em{linear neighborhood complexity}} if there exists a constant $c$ such that for every graph $G\in \Cc$ and a nonempty vertex subset $S\subseteq V(G)$, we have $|\{N(u)\cap S\colon u\in V(G)\}|\leq c\cdot |S|$; and $\Cc$ has {\em{inherently linear neighborhood complexity}} if every class transducible from $\Cc$ has linear neighborhood complexity. It follows from the results of Pilipczuk, Siebertz, and Toru\'nczyk~\cite{PilipczukST18a} that every class of structurally bounded expansion has inherently linear neighborhood complexity. Thus, the prerequisites of \cref{thm:main} hold for every class of structurally bounded expansion, but the result implies more: Since the conclusion entails that the original class $\Cc$ is transducible from a bounded expansion class $\Dd$, in fact the prerequisites --- monadic stability and inherently linear neighborhood complexity --- are {\em{equivalent}} to having structurally bounded expansion. That is, \cref{thm:main} implies the following.

\begin{corollary}\label{cor:main-structural}
	The following conditions are equivalent for a graph class $\Cc$:
	\begin{itemize}[nosep]
		\item $\Cc$ has structurally bounded expansion;
		\item $\Cc$ is monadically stable and has inherently linear neighborhood complexity; and
		\item $\Cc$ is monadically stable and has bounded merge-width.
	\end{itemize}
\end{corollary}

Here, merge-width is a family of graph parameters introduced recently by Dreier and Toru\'nczyk in~\cite{mergewidth}. The results of~\cite{mergewidth} together with the recent proof of Bonamy and Geniet~\cite{BonamyG25} that classes of bounded merge-width have linear neighborhood complexity imply that (i) every class of structurally bounded expansion has bounded merge-width, and (ii) every class of bounded merge-width has inherently linear neighborhood complexity. This justifies including the last item in \cref{cor:main-structural}.

\paragraph{Algorithmic implications.}
Let us briefly discuss some algorithmic consequences of \cref{thm:main}. The main motivation behind the question of Gajarsk\'y et al.~\cite{gajarsky2020sbe} was to give a fixed-parameter algorithm for the model-checking problem for first-order logic  on any class of structurally bounded expansion. Recall that this means that given an $n$-vertex graph $G$ from the class in question and a first-order sentence $\varphi$, one can decide whether $\varphi$ holds in $G$ in time $\Oh_\varphi(n^c)$ for some universal constant $c$, where the $\Oh_\varphi(\cdot)$ notation hides factors that may depend on~$\varphi$. The idea is that a polynomial-time sparsification procedure for a class of structurally bounded expansion $\Cc$ immediately allow us to reduce the $\mathsf{FO}$ model-checking problem on $\Cc$ to the setting of a bounded expansion class $\Dd$, where one can use the algorithm of Dvo\v{r}\'ak, Kr\'al', and Thomas~\cite{dvorak2013fo_be}. Indeed, if we are given a graph $G\in \Cc$ and we have that $G=I(H)$ for some fixed interpretation $I=(\psi,\nu)$ and a graph $H\in \Dd$, then testing whether a given first-order sentence $\varphi$ holds in $G$ is equivalent to testing whether a suitably ``pulled back'' sentence $\varphi'$ holds in $H$. (More precisely, $\varphi'$ is obtained from $\varphi$ by replacing every adjacency check by the formula $\psi$ and guarding every quantification with the formula $\nu$.)

By now, that $\mathsf{FO}$ model-checking is fixed-parameter tractable on every graph class of structurally bounded expansion follows from the more general results of~\cite{DreierMS23,stable_MC}, which establish this conclusion for structurally nowhere dense and monadically stable classes, respectively. The algorithms of~\cite{DreierMS23,stable_MC} circumvent the need of performing efficient sparsification\footnote{We note here that the sparsification result of Braunfeld et al.~\cite{horizons} heavily relies on the techniques developed in~\cite{stable_MC,DreierMS23,GajarskyMMOPPSS23} for working with the $\mathsf{FO}$ model-checking problem on monadically stable graph classes.} and work directly on the considered graph $G$. However, the sparsification route has still its merits, as our understanding of $\mathsf{FO}$ on graph classes of bounded expansion is much more robust than on monadically stable classes. In particular, $\mathsf{FO}$ on graph classes of bounded expansion admits a quantifier elimination procedure that can be executed in fixed-parameter time (see~\cite{dvorak2013fo_be,grokre11,PilipczukST18} for different presentations of this result). This allows obtaining algorithmic results for various problems connected to $\mathsf{FO}$ beyond basic model-checking, such as:
\begin{itemize}[nosep]
	\item counting solutions to $\FO$ queries and evaluation of various aggregation queries~\cite{KazanaS19,Torunczyk20}; 
	\item partially dynamic data structures supporting answering $\FO$ queries in constant time~\cite{dvorak2013fo_be,Torunczyk20};
	\item enumeration of $\FO$ queries with constant delay, also in a partially dynamic setting~\cite{KazanaS19,Torunczyk20}; and
	\item const.-factor approximation algorithms for monotone maximization problems expressible in $\FO$~\cite{Dvorak22}.
\end{itemize}
Similarly as for the model-checking problem,
\cref{thm:main} can be used to automatically lift all these results to the generality of any graph class of structurally bounded expansion, at the cost of an $\Oh(n^4)$-time preprocessing. We remark that none of these results has been worked out in the generality of structurally nowhere dense or monadically stable classes, and in fact they are expected to hold only with worse complexity guarantees (with all the constant upper bounds replaced by subpolynomial functions of the vertex count). So all the algorithmic consequences of \cref{thm:main} mentioned above are genuinely new~findings.

Finally, we would like to point out a useful feature of \cref{thm:main}: the sparsified graph class $\Dd$ can be transduced from the given class $\Cc$. This means that if we know something more about $\Cc$, for instance that it enjoys some property $\cal P$ that is closed under applying transductions, then $\Dd$ also enjoys $\cal P$. By substituting $\cal P$ for the property of having bounded linear cliquewidth, bounded cliquewidth, and bounded twin-width (which are all closed under transductions, see e.g.~\cite{michal_survey}), we conclude that:
\smallskip
\begin{itemize}[nosep]
	\item if $\Cc$ is a monadically stable class of bounded linear cliquewidth, then $\Dd$ has bounded pathwidth;
	\item if $\Cc$ is a monadically stable class of bounded cliquewidth, then $\Dd$ has bounded treewidth; and
	\item if $\Cc$ is a monadically stable class of bounded twin-width, then $\Dd$ has bounded sparse twin-width.
\end{itemize}
\smallskip
(Here we use the fact that classes that have both bounded expansion and bounded linear cliquewidth in fact have bounded pathwidth, and similarly for the other two notions; see~\cite{michal_survey}. Also, any class of bounded (linear) cliquewidth or bounded twin-width has inherently linear neighborhood complexity~\cite{BonnetKRTW22,Przybyszewski23}.) Since $\Cc$ can be transduced back from $\Dd$ using the interpretation $I$, we can immediately derive from \cref{thm:main} the following three results known in the literature:
\smallskip
\begin{itemize}[nosep]
	\item Every monadically stable graph class of bounded linear cliquewidth is transducible from a class of bounded pathwidth~\cite{nesetril2021linrw_stable}.
	\item Every monadically stable graph class of bounded cliquewidth is transducible from a class of bounded treewidth~\cite{DBLP:conf/soda/NesetrilMPRS21}.
	\item Every monadically stable graph class of bounded twin-width is transducible from a class of bounded sparse twin-width~\cite{GajarskyPT22}.
\end{itemize}

\paragraph{Technical contribution.}
Previous sparsification results~\cite{horizons,bushes_quasibushes} create a sparse representation of a graph by combining the \emph{Flipper game} of~\cite{GajarskyMMOPPSS23} and \emph{neighborhood covers}. A neighborhood cover of a graph $G$ is a collection \(\cal X\) of subsets of \(V(G)\) called \emph{clusters} such that for each \(v \in V(G)\) there is  \( X \in \cal X\) such that  \(N_G[v] \subseteq X \). We say that a neighborhood cover \(\cal X\) is \emph{compact} if there exists a partition \(\PP\) of \(V(G)\) such that (i) for each \(P \in \PP\) there is \(u \in V(G)\) with \(P \subseteq N_G[u]\) and (ii) for each \( X \in \cal X\) we have that \( X = N_G[P]\) for some \(P \in \PP\). Furthermore, we say that a neighborhood cover \(\cal X\) has \emph{overlap} at most \(k\) if every vertex \(v\) is contained in at most \(k\) clusters of \(\cal X\).

It is well-known that structurally bounded expansion classes admit compact neighborhood covers of constant overlap~\cite{coloring-covering},
and thus a non-algorithmic version of \Cref{thm:main} (for the special case that \(\C\) has structurally bounded expansion)
easily follows from existing work~\cite{horizons,bushes_quasibushes}.
However, so far, \emph{efficient computation of compact neighborhood covers} has been the central bottleneck towards proving the crucial algorithmic aspect of \Cref{thm:main}.
On graph classes that admit compact neighborhood covers with overlap~\(k\), we only knew how to compute compact neighborhood covers of overlap \(k\cdot \log^{\Oh(1)} n\)~\cite{DreierMS23,stable_MC},
which only allows us to compute sparsity parameters of order \(k\cdot \log^{\Oh(1)} n\).
This is insufficient for any of the above mentioned structural or algorithmic applications.
In this paper, we fill this gap:

\begin{restatable}{lemma}{lemmacovers}\label{lem:covers}
    Let \(\C\) be a monadically stable graph class with inherently linear neighborhood complexity.
   Then there is a constant \(k \in \N\) and an algorithm that computes for any \(G \in \C\) a compact neighborhood cover of overlap \(k\)  in time \(\mathcal O(|G|^4)\).
\end{restatable}

While not properly documented in the literature, this lemma,
in combination with prior work on bushes and flipper games, would allow an experienced reader to prove \Cref{thm:main}.
However, such a proof would rely on many nontrivial, external results such as the Flipper game~\cite{GajarskyMMOPPSS23},
which itself relies on further nontrivial results such as flip-flatness~\cite{dreier2022indiscernibles}.
In this paper, we take a much more direct route.
We replace the use of neighborhood covers with the following lemma, our core technical contribution.

\begin{lemma*}[simplified version of \Cref{lem:main}]
    Let $\C$ be a monadically stable class of bipartite graphs with linear neighborhood complexity. Then there exists an integer $k^*\in \N$ and an algorithm with running time $\mathcal{O}(|G|^4)$ that given any  $G = (A,B,E)$ from $\C$ without isolated vertices computes a partition $\F= \{ P_1,\ldots, P_{|\F|}\}$ of $B$ and vertices $\ell_1,\ldots, \ell_{|\F|}$ from $A$ such that:
    \begin{itemize}[nosep]
        \item $P_i \subseteq N_G(\ell_i)$ for each $i\in \{1,\ldots,|\F|\}$, and
        \item every $v \in A$ has a neighbor in at most $k^*$ parts of $\F$, i.e.,
        \(|\{P \in \F\ |\ N_G(v) \cap P \not= \emptyset\}| \le k^*\).
    \end{itemize}
Each vertex $\ell_i$ is called the \emph{leader} of part $P_i$.
\end{lemma*}

It is easily shown (\Cref{sec:mainlemma-implies-covers}) that  \Cref{lem:main} implies \Cref{lem:covers}.
The proof of \Cref{lem:main} itself is surprisingly short and self-contained.
Its only external dependencies are Haussler's Packing Lemma~\cite{haussler1995sphere} and a simple transitivity statement about the near-twin relation in graphs that exclude a half-graph~\cite[Lemma~23]{RoseLemma}.
To derive \Cref{thm:main} from \Cref{lem:main}, we replace the established Flipper game methods from the literature~\cite{horizons,bushes_quasibushes} with an iterated application of \Cref{lem:main}
and basic half-graph based~arguments.

\paragraph*{Organization.} In \cref{sec:overview} we present a brief technical overview of the proof of \cref{thm:main}, focusing on the key ideas. In \cref{sec:prelims} we establish terminology, recall key definitions concerning interpretations and transductions, and prove some auxiliary lemmas. \cref{sec:mainlemma} is devoted to the proof of the main technical statement, \Cref{lem:main}, that captures a single step of our sparsification procedure. Then, in \cref{sec:sparsify} we show how to iterate this lemma, thereby proving \cref{thm:main}.
 
Finally, in \cref{sec:mon-dep} we show that our method can be also used to recover the sparsification result of Braunfeld et al.~\cite{horizons}: by applying our procedure to any monadically stable graph class, we may sparsify it to an almost nowhere dense class (see \cref{thm:almost-nowhere-dense-sparsification} for a formal statement). 

\section{Technical overview}\label{sec:overview}

For technical reasons it will be convenient to work with bipartite graphs. 
This comes at no loss of generality, since given any graph $G$, one can compute a bipartite graph $G'$ such that if we can efficiently sparsify~$G'$, then we can also sparsify the original graph $G$. The construction of $G'$ from $G$ essentially boils down to copying every vertex and, for a pair of adjacent vertices $u,v$, making the first copy of $u$ adjacent to the second copy of $v$ and vice versa.
Consequently, for the rest of the section we will work with a monadically stable class $\C$ of bipartite graphs of inherently linear neighborhood complexity. Moreover, we can assume that $\C$ is hereditary and closed under bipartite complements, as both properties remain preserved if we close $\C$ under taking induced subgraphs and bipartite complements.

\paragraph{Description of the outcome.}
Our goal is to efficiently compute, for any $G \in \C$, a sparse graph $H$ such that $G = I(H)$ for some fixed interpretation $I$ depending only on $\C$. The graph $H$ that we will compute will have a special form, consisting of two parts: it will contain all the vertices of the original graph $G$, but without any edges connecting them, and a rooted tree $T$ to which the vertices of $G$ will be suitably connected.
Crucially, the tree $T$ will be of bounded height (by a bound depending only on $\C$). 
In what follows we will refer to the vertices of $T$ as \emph{nodes}.
We will show that each \(v \in V(G)\) is adjacent only to a bounded number of nodes from \(T\) (by a bound $d$ depending only on $\C$). This will guarantee that $H$ is sparse --- in particular, it is $d$-degenerate (i.e. every subgraph has a vertex of degree at most $d$). Consequently, the graph class $\D$ consisting of all graphs $H$ computed by our sparsification procedure, is $d$-degenerate. Being degenerate is a notion of sparseness that is considerably weaker than being of bounded expansion; however, we will be able to show that $\D$ in fact has bounded expansion.

\paragraph{Recovering $G$ from $H$.}
We now describe how one can recover the original graph $G$ from $H$. The construction of $T$ and $H$ will guarantee that if $u,v$ are any two vertices of $G$ in $H$, then there  always exists a node $p$ such that both $u$ and $v$ are adjacent to $p$, but no child $p'$ of $p$ has this property. To determine whether $u$ and $v$ are adjacent in $G$, we simply count the number of nodes on the path from $p$ to the root of $T$. If this number is even, then they are adjacent, and otherwise they are non-adjacent (see also \Cref{fig:sparse_graph}). This property, together with the fact that $T$ is of bounded height, allow us to recover the adjacency of $G$ from $H$ by a first-order formula. We remark that while for any $u,v \in V(G)$ there may exist several distinct nodes of $T$ with the property described above, the construction of $T$ and $H$ will guarantee that performing the parity check for any such node will lead to the same result.

We remark that the form of the graph $H$ described above as well as the mechanics of recovering $G$ from $H$ are heavily inspired by the notions of \emph{bushes} and \emph{quasi-bushes} proposed by Dreier et al. in~\cite{bushes_quasibushes}. It would be possible to adjust our algorithm to produce bushes or quasi-bushes adhering to the definitions from~\cite{bushes_quasibushes}, however, the somewhat simpler form of the graph $H$ described above is more suitable for the arguments used in our proofs.

\begin{figure}[t]
  \begin{minipage}[t]{0.5\textwidth}
    \vspace*{0pt}

    \includegraphics[scale=1.1]{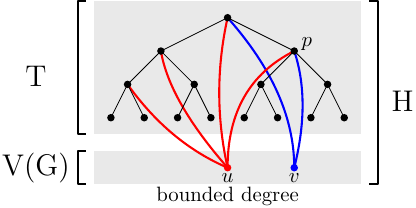}

  \end{minipage}\hfill
  \begin{minipage}[t]{0.5\textwidth}
    \vspace*{11pt}
    \captionof{figure}{
        For given vertices \(u,v \in V(G)\), we display a node \(p\) adjacent to both in \(H\), such that no child of \(p\) is adjacent to both.
        There are an even number of vertices on the path from the root to \(p\), and thus \(u,v\) are adjacent in \(G\).
    }
    \label{fig:sparse_graph}
  \end{minipage}
\end{figure}

\paragraph{Construction of $T$ and $H$.}

The key ingredient behind the construction of \(T\) and \(H\) is the following simple consequence of our main technical statement, \Cref{lem:main}.

Given any $G=(A,B,E)\in \Cc$ without isolated vertices, we can efficiently compute a collection of induced subgraphs $G_1,\ldots, G_m$  of $G$ (here $m$ is not bounded) satisfying the following conditions.
\smallskip
\begin{enumerate}[nosep]
    \item Each $v \in V(G)$ is in at least one and in at most $k^*$ graphs from $G_1,\ldots, G_m$, where $k^*$ is an integer that depends only on $\C$.
    \item Every edge of $G$ is contained in at least one $G_i$.
    \item Depending on the direction we apply \Cref{lem:main}, one of the following will be enforced:

    \begin{enumerate}
        \item \label{ca} All graphs $G_i$ contain a vertex from \(A \cap V(G_i)\) that is adjacent to all vertices from \(B \cap V(G_i)\).
        \item \label{cb} All graphs $G_i$ contain a vertex from \(B \cap V(G_i)\) that is adjacent to all vertices from \(A \cap V(G_i)\).

    \end{enumerate} 
\end{enumerate}
\smallskip
We will refer to this collection of induced subgraphs $G_1,\ldots, G_m$ of $G$ as a  \emph{left branching} or a \emph{right branching} of $G$, depending on whether we enforce~(\ref{ca}) or~(\ref{cb}).

From \Cref{lem:main} presented earlier, 
one can obtain a left branching of $G$ by setting 
$m\coloneqq |\F|$,
$G_i\coloneqq G[N(P_i),P_i]$,
and choosing the dominating vertex of (\ref{ca}) to be \(\ell_i\).
To obtain a right branching, we simply exchange the roles of $A$ and $B$.

Let us now describe the tree $T$. 
As part of the construction, each node $p$ will be assigned a graph $G(p)$ that is either an induced subgraph of $G$ (if $p$ is at even depth),
or the bipartite complement of an induced subgraph of $G$ (if $p$ is at odd depth). 
To construct \(T\), we start with creating a root node $r$ and we set $G(r) \coloneqq G$.
Then we remove all isolated vertices from $G(r)$ and compute a left branching $G_1,\ldots, G_m$ of the resulting graph.
For each branch $G_i$, we create a new child $p_i$ of $r$, and set $G(p_i)$ to be the bipartite complement of $G_i$.
We then recurse into each child $p_i$ of $r$, and proceed analogously for $G(p_i)$, again removing all isolated vertices of each $G(p_i)$, but this time using a \emph{right} branching of the resulting graph. For each branch, we again create a node and again assign the bipartite complement of the branch to this child. To summarize, the general construction looks as follows: For any node $p$ and graph $G(p)$, we remove from $G(p)$ all isolated vertices (if there are any) and compute a left or right branching $G_1,\ldots,G_m$ of $G(p)$, depending on whether $p$ is at even or odd depth. Then for each branch $G_i$ we create a new child $p_i$ of $p$, and set $G(p_i)$ to be the bipartite complement of $G_i$ and proceed recursively into each $p_i$. The recursion stops when $G(p)$ is an edgeless graph.

For any $G\in \Cc$, we define the {\em{ladder index}} of $G$ as the maximum order of a half-graph that is contained in $G$ as an induced subgraph. Since $\Cc$ is monadically stable, there is a finite bound $t$, depending only on~$\Cc$, on the ladder index of graphs from $\Cc$. The ladder index is a convenient parameter to measure the progress of our algorithm: the alternation between induced subgraphs of $G$ and complements of induced subgraphs of $G$ in our choice of $G(p)$, as well as between using left branchings and right branchings, leads to the decrease of the ladder index in every second step of the construction. To be more precise -- it is not difficult to show that if $p$ and $p''$ are nodes of $T$ such that $p$ is the grandparent of $p''$ in $T$, then the ladder index of $G(p'')$ is smaller than the ladder index of $G(p)$.
As the ladder index of any graph from $\C$ is bounded by $t$, the height of $T$ is hence bounded by $2t$.

Once the tree $T$ is computed from $G$, we finish the construction of $H$ by 
adding the vertex set of $G$ to $T$ (but without any edges),
and for every $v \in V(G)$ connecting $v$ to all nodes $p$ of $T$ with $v \in V(G(p))$.
The bound on the depth of \(T\), and on the branches per node that each vertex is contained in, ensure that each \(v \in V(G)\) has bounded degree.

\paragraph{Proving that $\D$ has bounded expansion.}
The proof that $\D$ is of bounded expansion can be split into three parts. 
First, one argues that $\D$ is weakly sparse, that is, that there exists a number $h$ such that no $H \in \D$ contains $K_{h,h}$ as a subgraph. This is an easy consequence of the fact that $\D$ is $d$-degenerate for some $d$ depending only on $\C$, as then we can set $h\coloneqq d$.

Second, one argues that $\D$ is transducible from $\C$. This is technically the most challenging part of the proof, but ultimately follows from the fact we can not only compute a (left or right) branching of the input graph efficiently, but the branching computed by our algorithm is also definable by a first-order formula in a suitable sense. 

The third step uses the previous two steps and combines them with a result of Dvo\v{r}\'ak~\cite{Dvorak18} to argue that $\D$ has bounded expansion. This result states that if $\D$ is a weakly sparse graph class, then to show that $\D$ has bounded expansion it is enough to prove that for every $r \ge 1$, graphs from $\D$ do not contain arbitrarily dense graphs as induced $({\le}r)$-subdivisions. 
We use this as follows.
As $\D$ is transducible from~$\C$, we know that it also has inherently linear neighborhood complexity. Now we proceed by contradiction and assume that for some fixed $r \ge 1$ the graphs from $\D$ contain arbitrarily dense graphs as induced $({\le}r)$-subdivisions. From this and the assumption that $\D$ is weakly sparse, we can argue that there is a transduction $\T$ that from $\D$ transduces a graph class that contains arbitrarily dense graphs as $1$-subdivisions. It is easy to see that such a graph class cannot have linear neighborhood complexity, which leads to a contradiction.

\paragraph{Main lemma.}

It remains to discuss the idea behind the proof of the main lemma, \Cref{lem:main}.
We will rely on the following consequence of Haussler's packing lemma (a combinatorial result about packings in set systems): Let $\C$ be a class of bipartite graphs with linear neighborhood complexity. Then there exists $k$  such that if $G = (A,B,E)$ is a bipartite graph from $\C$ that has no twins in $B$, then there are two vertices $u,v$ in $A$ that are $k$-near-twins. Here, by $u,v$ being {\em{$k$-near-twins}} we mean that their adjacency towards $B$ agrees on all but at most $k$ vertices. We remark that the number $k$ used in this statement is different from the value $k^*$ claimed in the main lemma, but we will see that $k^*$ is closely related to $k$.

We now describe the algorithm to compute $\F$.
Let $n\coloneqq |A|$.
We set $G_0 \coloneqq G$ and proceed in $n$ rounds as follows. To obtain $G_i$ from $G_{i-1}$, we first eliminate all twins on the right side of $G_{i-1}$ (by keeping exactly one vertex from each twin class and deleting the rest) to obtain graph $G_{i-1}'$. Then we find a pair $u,v$ of  $k$-near-twins of $G_{i-1}'$ and remove one of them to obtain $G_i$. We also remember the removed vertex as~$a_i$. By proceeding in this way, we eventually eliminate the whole graph and obtain an ordering $a_1,\ldots, a_n$ of~$A$. 
We then define the partition $\F$ claimed in the main lemma as follows: For each $i \in [n]$, let $S_i$ be the set of vertices $v$ from $B$ such that $i$ is the largest index for which $a_i$ is a neighbor of $v$. We then define $\F$ to be the set of all nonempty sets $S_i$, and choose $a_i$ as the leader of $S_i$. 

To show that $\F$ has the desired properties, we proceed as follows. It is easy to see that each part of $\F$ is contained in the neighborhood of its leader, so we only need to show that each vertex from $A$ has neighbors in at most $k^*$ parts of $\F$. We define a bipartite graph $G(A,\F)$ with sides $A$ and $\F$ in which there is an edge $vP$ for $v \in A$ and $P \in \F$ if $v$ has a neighbor in $P$. To finish the proof of the main lemma, we will show that every vertex from $A$ has at most $k^*$ neighbors in $G(A,\F)$. To this end, we establish the following two properties of $G(A,\F)$:
\begin{enumerate}[(i)]
    \item The ladder index of $G(A,\F)$ is bounded by a constant depending only on $\C$. Note that this is nontrivial: while $G$ cannot contain a large half-graph as an induced subgraph because $\C$ is monadically stable, it is not immediately clear that $G(A,\F)$ also should not contain a large half-graph.
    \item \label{p2} For every $a_i$ with $i<n$, there exists $a_j$ with $j > i$ that is a $k$-near-twin of $a_i$ in $G(A,\F)$. 
\end{enumerate}

Then we consider a graph $K$ with vertex set $V(G)$ and edges between all pairs of vertices that are $k$-near-twins in $G(A,\F)$. By property (ii) above there is a path from each $a_i$ to $a_n$ in $K$. In particular, all vertices from $A$ are in the same connected component of $K$. We then invoke the result from~\cite{RoseLemma} that states if $G$ is a graph of ladder index at most $t$ and $K$ is a $k$-near-twin graph of $G$, then any two vertices that are in the same connected component of $K$ are $k'$-near-twins of $G$, for some $k'$ depending on $t$ and~$k$. 
By the construction of $\F$ we see that $a_n$ has degree at most $1$ in $G(A,\F)$, and therefore all vertices in $A$ have degree at most $k^*\coloneqq  k' + 1$ in $G(A,\F)$. This finishes the proof outline of the main lemma.

We remark that the presentation of the algorithm that computes the partition $\F$ differs from the presentation given in Section~\ref{sec:mainlemma}. The presentation given here is more concise, while the presentation in Section~\ref{sec:mainlemma} is more suitable for the arguments used in the proofs. However, it is not difficult to see that the two formulations of the algorithm are equivalent.

\section{Preliminaries}\label{sec:prelims}

For a nonnegative integer $k$, we denote $[k]\coloneqq \{1,\ldots, k\}$.

\subsection{Graphs}
All our graphs are finite, simple and undirected.
We use standard graph terminology and notation, some of which we will now briefly review. For a graph $G$ and a vertex $v \in V(G)$, we denote by $N_G(v)$ the set of neighbors of $v$ in $G$. If $S$ is a subset of $V(G)$, we denote by $G[S]$ the subgraph of $G$ induced by $S$.
We define \(|G| \coloneqq |V(G)|\) and \(\|G\| \coloneqq |E(G)|\).
By a \emph{rooted tree} we mean a connected acyclic graph in which there is a designated vertex, called its \emph{root}. The \emph{height} of a tree $T$ is the largest number of edges on any leaf-to-root path in $T$. The \emph{depth} of a vertex $v$ in a tree $T$ is the number of edges on the unique path from $v$ to the root of $T$.

Let $H$ be a graph. An \emph{$({\le} r)$-subdivision} of $H$ is any graph that can be obtained from $H$ by replacing each edge $uv$ of $H$ by a path connecting $u$ and $v$ with at most \(r\) internal vertices (possibly $0$, which means that the edge stays intact).

We call the original vertices of $H$ the \emph{principal} vertices of the $({\le} r)$-subdivision and the internal vertices of the paths the \emph{subdivision} vertices.

Our main technical lemmas will work with bipartite graphs.
For a bipartite graph $G = (A,B,E)$ we define $L(G) = A$ and $R(G)= B$. For $X \subseteq A$ and $Y \subseteq B$, we define $G[X,Y]$ 
to be the graph $(X,Y,E \cap (X\times Y))$. 
By the \emph{bipartite complement} of $G= (A,B,E)$ we mean the graph $\widetilde{G} \coloneqq (A,B,E')$ where $E'\coloneqq \{uv \in A\times B~|~uv\not\in E\}$.

\subsection{Logic}

\paragraph{Logic and graphs.} 
We assume familiarity with first-order logic $\FO$.

We will work with graphs, modelled as structures over a signature consisting of one binary predicate $\mathsf{adj}$, indicating adjacency between vertices. The vertices of our graphs will often be marked with (possibly overlapping) colors, each modelled by an additional unary predicate. We then speak of \emph{colored} graphs and if a colored $G'$ is obtained from $G$ by adding some unary predicates, then $G'$ is a {\em{unary coloring}} of $G$.
Graphs that are explicitly bipartite are equipped with additional unary predicates \(L\) and \(R\) to distinguish their two sides.

\paragraph{Transductions and interpretations.}

Let $k \in \N$ and let $G$ be a graph. A \emph{$k$-copy} of $G$ is the structure $G^+$ over signature $\{\mathsf{adj}, \mathsf{cp}\}$  with universe $V(G)\times [k]$ in which the binary relations $\mathsf{adj}_{G^+}$ and $\mathsf{cp}_{G^+}$ are realized as follows: We have that $(u,i)(v,j) \in \mathsf{adj}_{G^+}$  if and only if $i=j$ and $uv \in E(G)$ and we have $(u,i)(v,j) \in \mathsf{cp}_{G^+}$ if and only if $u=v$. Informally, $G^+$ consists of $k$ disjoint copies of $G$ in which we can determine whether two vertices correspond to the same vertex of the original graph by predicate $\mathsf{cp}$.

A first-order \emph{transduction} from graphs to graphs is determined by a number $k$, a finite set $\Lambda$ of unary predicates and a first-order formula $\psi(x,y)$ over the signature $\{\mathsf{adj}, \mathsf{cp}\} \cup \Lambda$. If $\T = (k, \Lambda, \psi)$ is a transduction and $G$ and $H$ are graphs, we say that $\T$ \emph{transduces} $H$ from $G$, denoted by $H \in \T(G)$, if $H$ can be obtained from $G$ via the following steps:
\begin{enumerate}[nosep]
    \item First, a $k$-copy $G^+$ of $G$ is created.
    \item Then we mark vertices of $G^+$ with unary predicates from $\Lambda$ arbitrarily, thus obtaining a unary coloring $G^*$ of $G^+$.
    \item We define graph $\psi(G^*)$ as the graph on vertex set $V(G^*)$ in which there is an edge between $u,v$ if and only if $G^* \models \psi(u,v)$.
    \item Finally, as the resulting graph $H$ we take an arbitrary induced subgraph of $\psi(G^*)$.
\end{enumerate}

If $\C$ is a class of graphs, we set $\T(\C)\coloneqq\bigcup_{G \in \C}\T(G)$. We then say that a graph class $\D$ is {\em{transducible}} from $\C$ if $\D \subseteq \T(\C)$ for some transduction $\T$.

Transductions can be composed -- if $\ST$ and $\T$ are two transductions, we denote their composition by $\ST \circ \T$, which is again a transduction. We refer to~\cite{gajarsky2020sbe,michal_survey} for a more detailed overview of transductions and their properties. 

A first-order \emph{interpretation} from graphs to graphs consists of two formulas $\psi(x,y)$ and $\nu(x)$. If $I=(\psi,\nu)$ is an interpretation and $G$ is a graph, we define graph $I(G)$ as follows. The vertex set of $I(G)$ is the set $\{ v \in V(G)\ |\ G \models \nu(v)\}$ and the edge set of $I(G)$ is defined as $\{ uv \in V(G)\times V(G)\ |\ G \models \psi(u,v)\}$.

\subsection{Graph classes}

A {\em{graph class}} is just a set of graphs, typically infinite. We need the following standard definitions.

\begin{definition}
    We say that a graph class $\C$ is \emph{weakly sparse} if there exists $h\in \N$ such that no graph $G \in \C$ contains $K_{h,h}$ as a subgraph. 
\end{definition}

\begin{definition}
    We say that a graph $G$ is $d$-degenerate if every subgraph of $G$ contains a vertex of degree at most $d$. We say that a class $\C$ of graphs is degenerate if there exists $d$ such that each $G \in \C$ is $d$-degenerate.
\end{definition}
It is well-known that for every $d$-degenerate graph there exists a proper coloring of its vertices with $d+1$ colors. Also, it is easily seen that every degenerate graph class is weakly sparse.

\begin{definition}
A \emph{half-graph of order $t$}, denoted \(H_t\), is a graph with vertex set $\{u_1,\ldots,u_t, v_1,\ldots, v_t\}$ and with edge set $\{u_iv_j\ \colon\ 1\le i \le j\le t\}$.

\end{definition}
We note that a half-graph of order $t$ is often also called a \emph{ladder} of order $t$ in the literature.

\begin{remark}\label{rem:ladders_sides}
	We will often work with half-graphs in the context of bipartite graphs. It is easy to see that if a bipartite graph $G$ with sides $A$ and $B$ contains $H_t$ with vertices $v_1,\ldots, v_n,u_1,\ldots,u_n$ as an induced subgraph, then we have that $v_1,\ldots, v_n \in A$ and $u_1,\ldots,u_n \in B$ or we have that $v_1,\ldots, v_n \in B$ and $u_1,\ldots,u_n \in A$. In some arguments it will be convenient for us to pick which of these two options is true. This does not cause a problem, since if $v_1,\ldots, v_n \in A$ and $u_1,\ldots,u_n \in B$, then we may rename the vertices suitably so that $v_1,\ldots, v_n \in B$ and $u_1,\ldots,u_n \in A$ holds, i.e. we swap the names of $v_1$ and $u_n$, $v_2$ and $u_{n-1}$~etc.
\end{remark}

\begin{definition}
    We say that a graph class $\C$ is \emph{monadically stable} if the class of all half-graphs is not transducible from $\C$.
\end{definition}

\begin{definition}
    We say that a class $\C$ of graphs has \emph{linear neighborhood complexity} if there exists $c$ such that for every $G \in \C$ and every nonempty $S \subseteq V(G)$ we have that 
    $$ |\{N(v) \cap S~:~v\in V(G) \}| \le c\cdot|S|.$$
    We say that \(\C\) has \emph{inherently linear neighborhood complexity} if 
    for every transduction $\T$, the class $\T(\C)$ has linear neighborhood complexity.
\end{definition}

We say that two vertices $u,v$ of a graph $G$ are \emph{$k$-near-twins} if $|N_G(v) \triangle N_G(u)| \le k$, i.e. if they have the same neighborhoods with at most $k$ exceptions.
The following claim is a simple consequence of Haussler's Packing Lemma~\cite{haussler1995sphere}.
By translating the Packing Lemma from the language of set systems (used by Haussler) to the language of bipartite graphs
(as done e.g. in~\cite[Lemma 3.1]{davies2025oddcoloringgraphslinear}), one obtains the following.

\begin{lemma}\label{lem:bip-graph-near-twins}
Let $\C$ be a hereditary class of bipartite graphs with linear neighborhood complexity. 
Then there exists a number $k$ such that the following holds: 
For every graph $G = (A,B,E)\in \C$ that contains no pair of twins in $B$ and satisfies $|A|\geq 2$, the set $A$ contains a pair of $k$-near-twins. 
\end{lemma}

Finally, the notion of
bounded expansion was introduced by Ne\v{s}et\v{r}il and Ossona de Mendez~\cite{nevsetvril2008grad}. Instead of using the original definition, we will work with the following characterization due to Dvořák~\cite{Dvorak18}.

\begin{lemma}[\cite{Dvorak18}]\label{lem:BE_subdivisions}
    A class of graphs $\D$ has bounded expansion if and only if $\D$ is weakly sparse and for every $r \in \N$ there exists $d \in \N$ such that the following holds: If $H$ is a graph such that some $G \in \D$ contains an $({\le} r)$-subdivision of $H$ as an induced subgraph, then $\frac{\|H\|}{|H|} \le d$.
\end{lemma}

\newcommand{\kHaussler}{k^{\circ}}

\section{Main lemma}
\label{sec:mainlemma}

In this section we prove our main technical lemma, stated below.
Its output is visualized in \Cref{fig:main_lemma}.

\begin{figure}[t]
  \begin{minipage}[t]{0.52\textwidth}
    \vspace*{0pt}
    \begin{center}
    \includegraphics[scale=1.1]{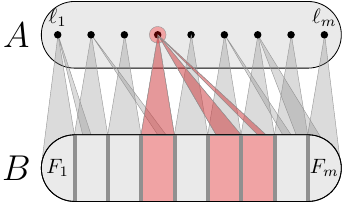}
    \end{center}
  \end{minipage}\hfill
  \begin{minipage}[t]{0.44\textwidth}
    \vspace*{25pt}
    \captionof{figure}{
    Visualization of the output of \Cref{lem:main}.
    We can choose \(k^* = 3\) and verify that, for example, the red vertex is adjacent to at most 3 parts from \(\F\), also marked red.
    }
    \label{fig:main_lemma}
  \end{minipage}
  \vspace{-0.1cm}
\end{figure}

\begin{lemma}
\label{lem:main}
    Let $\C$ be a monadically stable class of bipartite graphs with linear neighborhood complexity. Then there exists a number $k^*$ and an algorithm with running time $\mathcal{O}(|G|^4)$ that for any given graph $G = (A,B,E)$ from $\C$ without isolated vertices computes a partition $\F= \{ P_1,\ldots, P_{|\F|}\}$ of $B$ and vertices $\ell_1,\ldots, \ell_{|\F|}$ belonging to $A$ such that:
    \begin{itemize}[itemsep=1pt]
        \item $P_i \subseteq N_G(\ell_i)$, for each $i\in \{1,\ldots,|\F|\}$; and
        \item every $v \in A$ has a neighbor in at most $k^*$ parts of $\F$, i.e.,
        \(|\{P \in \F\ |\ N_G(v) \cap P \not= \emptyset\}| \le k^*.\)
    \end{itemize}
    Moreover, the partition $\F$ of $B$ is definable in the following sense: 
    There exists a first-order formula $\mathsf{leader}(x,y)$ depending only on $\C$,
    and a graph \(G^+\) obtained from \(G\) by marking vertices with unary predicates such that 
    for any $u \in A,v \in B$
    we have
    $ G^+ \models \mathsf{leader}(u,v)$
    if and only if $u= \ell_i$ and $v \in P_i$ for some $i \in \{1,\ldots, |\F|\}$.
\end{lemma}

For the rest of the section we fix a monadically stable graph class $\C$ of bipartite graphs with linear neighborhood complexity and work only with graphs from this class. It will be convenient to assume that $\C$ is hereditary; this assumption comes at no cost of generality, since both properties of being monadically stable and having linear neighborhood complexity are preserved if we close our graph class under induced subgraphs.

\subsection{The algorithm}\label{sec:algo-main-lemma}

Before proceeding with the description of the algorithm we establish some terminology and notation.
For every set \(A' \subseteq A\) and collection \(\B\) of disjoint subsets of \(B\), we define a bipartite graph
$G(A', \B)$ with sides $A'$ and $\B$ in which there is an edge between $v \in A'$ and $P \in \B$ if and only if $v$ has a neighbor in $P$.

Let \(A' \subseteq A\).
Assume now that $\B$ is a partition of \(B\) into twin classes in the graph $G[A',B]$ and observe the following useful properties of $G(A',\B)$.
First, the graph $G(A',\B)$ fully encodes the connections between \(A'\) and \(B\) in \(G\).
More precisely,
if there is an edge between $v \in A$ and $P \in \B$, then \(v\) is adjacent to all the vertices of \(P\) in \(G\); and 
if there is no edge between $v \in A$ and $P \in \B$, then \(v\) is adjacent to no vertices of \(P\) in \(G\).
Second, $G(A',\B)$ is (isomorphic to) an induced subgraph of $G$, because it can be obtained
from $G[A',B]$ by keeping exactly one vertex
from each part of $\B$. Lastly, note that there are no $P,P' \in \B$ that are
twins in $G(A',\B)$.
Consequently, by \Cref{lem:bip-graph-near-twins}, for some \(\kHaussler\) depending only on \(\C\),
there exist $u,v \in A'$ that are $\kHaussler$-near-twins in $G(A',\B)$.
Throughout this chapter, we consistently denote the near-twin constant of \Cref{lem:bip-graph-near-twins} by \(\kHaussler\).

\begin{figure}[t]
  \begin{minipage}[t]{0.52\textwidth}
    \vspace*{0pt}
    \begin{center}
    \includegraphics[scale=1.1]{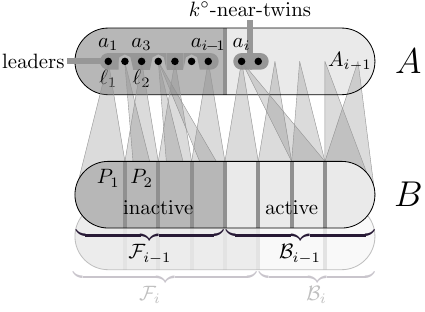}
    \end{center}
  \end{minipage}\hfill
  \begin{minipage}[t]{0.44\textwidth}
    \vspace*{5pt}
    \captionof{figure}{Visualization of step \(i\) of the algorithm.
      As depicted, each set \(P_j\) fully adjacent to \(a_j\), and fully non-adjacent to \(a_{j+1},a_{j+2}\), etc.
      Vertex \(a_i\) and its highlighted partner are \(1\)-near-twins, disagreeing only on the first set (from the left) of \(B_{i-1}\).
      As softly depicted at the bottom, \(\mathcal{B}_i\) is obtained from \(\mathcal{B}_{i-1}\) by removing the first set (as \(a_i\) was its only neighbor), and merging the second and third set (as \(a_i\) was its only distinguishing vertex).}
    \label{fig:algorithm}
  \end{minipage}
\end{figure}

We now proceed with the description of the algorithm, visualized in \Cref{fig:algorithm}.
Let $G = (A,B,E)$ be an input graph from $\C$. 
We set $n\coloneqq|A|$. In the algorithm we will compute a sequence of sets $A = A_0 \supset A_1 \supset \ldots  \supset A_{n} = \emptyset$, meaning that in each step we will remove precisely one vertex from $A_{i-1}$ to obtain $A_{i}$; this vertex will be called $a_{i}$. 

The output of the algorithm will be a partition $\F$ of $B$ and vertices $\ell_1,\ldots, \ell_{|\F|}$ as stipulated by the statement of \Cref{lem:main}, together with an ordering $a_1,\ldots, a_n$ of $A$; this ordering will be useful in the proof that $\F$ and $\ell_1,\ldots, \ell_{|\F|}$ have the required properties.

We initialize $A_0\coloneqq A$ and $\F_0\coloneqq \emptyset$. For $i\coloneqq 1$ to $n$ we do the following:
\begin{enumerate}
    \item 
        We say that a vertex $v$ in $B$ is \emph{active} over a set \(A' \subseteq A\) if $v$ has a neighbor in $A'$, otherwise we say that $v$ is \emph{inactive}.
        Let $\B_{i-1}$ be the partition of active vertices of $G[A_{i-1},B]$ into twin classes with respect to $A_{i-1}$. Find two vertices $u,v \in A_{i-1}$ that are $\kHaussler$-near-twins in $G(A_{i-1},\B_{i-1})$. Set $a_{i}\coloneqq u$ and $A_{i}\coloneqq A_{i-1} \setminus \{u\}$.
    \item If $P \in \B_{i-1}$ is a part such that $a_i$ is the only neighbor of $P$ in $G(A_{i-1},\B_{i-1})$, then set $\F_{i}\coloneqq \F_{i-1} \cup \{P\}$. In this case we say that $P$ is \emph{frozen} at time $i$ and say that $a_i$ is the \emph{leader} of $P$. Otherwise, we set $\F_{i}\coloneqq \F_{i-1}$.
\end{enumerate}
Let $\F\coloneqq \F_n$.
We order $\F = \{P_1,\ldots,P_{|\F|}\}$ ascending by the time of freezing, and
denote the corresponding leaders as $\ell_1,\ldots, \ell_{|\F|}$.

The running time of the algorithm can be bounded by $\Oh(|G|^4)$ -- the algorithm has $n \le |G|$ iterations, and finding a pair of $\kHaussler$-near-twins in each iteration can be done in cubic time. All other operations within any iteration can be easily performed in cubic time. 

We first observe that $\F=\F_n$ is indeed a partition of $B$.

\begin{lemma}\label{lem:partition}
    For every $i \in \{0,\ldots, n\}$, $\F_i$ is a partition of the inactive vertices over $A_i$. In particular, $\F_n$ is a partition of $B$. 
\end{lemma}

\begin{proof}
    By induction on $i$. For $i=0$ there are no inactive vertices (since we assumed that $G$ does not have isolated vertices) and so the empty partition $\F_0$ has the required property. 
    For $i>0$, if any vertices become inactive in the $i$-th iteration (that is, if there exists $P$ that only has $a_i$ as a neighbor in $G(A_{i-1},\B_{i-1})$), then the part containing all these vertices is added to $\F_i$. Otherwise the set of inactive vertices remains the same when going from $i-1$ to $i$, and then the algorithm uses the same partition.
\end{proof}

\begin{remark}
    We note that there is an alternative, more static definition of  $\F$ from which it is immediately clear that it is a partition of $B$. Let $a_1,\ldots, a_n$ be the ordering of vertices from $A$ as given by the algorithm. Then for each $i \in [n]$, let $S_i$ consist of all vertices $v \in B$ for which $va_i \in E(G)$ but  $va_j \not\in E(G)$ for all $j > i$ (in other words, $i$ is the largest index for which $a_i$ is a neighbor of $v$). We then set $\F\coloneqq \{S_i\ :\ S_i\not=\emptyset\}$ and for each $S_i \in \F$ we say that $a_i$ is its leader.
\end{remark}

Since $\F$ is a partition, and by construction we have for every $P_i \in \F$ that $P_i \subseteq N_G(\ell_i)$, the first assertion of the lemma is proven. The rest of this section is therefore devoted to the proof of the second assertion of the lemma, captured in the following statement.

\begin{lemma}
\label{lem:bd_degree}
    There exists $k^*$ depending only on $\C$ such that in the graph $G(A,\F)$, every vertex of $A$ has degree at most $k^*$.  
\end{lemma}

\subsection{Proof of \Cref{lem:bd_degree}}\label{sec:lem:bd_degree}

\Cref{lem:bd_degree} is a simple consequence of the three lemmas given below. We postpone the technical proofs of \Cref{lem:near_twins,lem:quotient_no_halfgraph} to \Cref{subsec:near_twins,subsec:no_halfgraph}, respectively, and first show how they imply \Cref{lem:bd_degree}.

Recall that $\kHaussler$ was the constant of \Cref{lem:bip-graph-near-twins} used by the algorithm.
\begin{lemma}\label{lem:near_twins}
	Let $G \in \C$ and let $a_1, \ldots, a_n$ be the ordering of vertices from $A$ computed by the algorithm. Then for every $i < n$ there exists $j > i$ such that $a_j$ is a $\kHaussler$-near-twin of $a_i$ in $G(A,\F)$.
\end{lemma}

\begin{lemma}\label{lem:quotient_no_halfgraph}
	There exists $t\in \N$ such that for every $G \in \C$ with $G=(A,B,E)$, the graph $G(A,\F)$ does not contain a half-graph of order $t$ as an induced subgraph.
\end{lemma}

The third lemma is due to Gajarsk\'y and McCarty~\cite{RoseLemma}. To state it, we need some terminology. For every graph $G$ and a non-negative integer $k$, we define the \emph{$k$-near-twin graph} of $G$ to be the graph on vertex set $V(G)$ in which there is an edge between vertices $u,v$ if and only if $u,v$ are $k$-near-twins in $G$.
\begin{lemma}[{\cite[Lemma 23]{RoseLemma}}]
\label{lem:Rose_lemma}
There is a function $h\colon \N^2 \to \N$
so that for any $k,t \in \N$, if $G$ is a graph with no half-graph of order $t$ as a semi-induced
\footnote{The precise definition of semi-induced subgraph is not important for us, but can be found for example in~\cite{GajarskyPT22}. We just note that for every $t \ge 1$ and every bipartite graph $G$ the following is true: $G$ contains $H_t$ as a semi-induced subgraph if and only if $G$ contains $H_t$ as an induced subgraph.}

subgraph, and $u$ and $v$ are vertices in the same component of the $k$-near-twin graph of $G$, then $u$ and $v$ are $h(k,t)$-near-twins in $G$.
\end{lemma}

\begin{proof}[Proof of \Cref{lem:bd_degree}]
	Let $H$ be the $\kHaussler$-near-twin graph of $G(A,\F)$. By \Cref{lem:near_twins}, every vertex $a_i$ with $i<n$ is adjacent to some $a_j$ with $j > i$ in $H$. Consequently, from every $a_i$ there is a path to $a_n$ in $H$, and so all the vertices from $A$ are in the same connected component $C$ of $H$. Since by \Cref{lem:quotient_no_halfgraph} the graph $G(A,\F)$ does not contain a half-graph of order $t$ as an induced subgraph, by \Cref{lem:Rose_lemma} all the vertices in $C$ are pairwise $h(\kHaussler,t)$-near-twins in $G(A,\F)$. Since $a_n$ has degree $1$ in $G(A,\F)$ (as one easily checks from the construction of $G(A,\F)$), each vertex in $C$ has degree at most $h(\kHaussler,t)+1$ in $G(A,\F)$. Since each vertex in $A$ is in $C$, the statement of the lemma follows.
\end{proof}

\subsubsection{Proof of \Cref{lem:near_twins}}
\label{subsec:near_twins} 

Recall that the algorithm computes sets $A_0 \supset A_1 \supset \ldots \supset A_n$ and for each $i$ we considered the partition $\B_i$ of the active vertices over $A_i$.
By \Cref{lem:partition} each $\F_i$ is a partition of the inactive vertices over $A_i$.
We can therefore define for each $i$ a partition $\PP_i$ of $B$ by setting $\PP_i\coloneqq \B_i \cup \F_i$. Then we have the following.
\begin{lemma}\label{lem:coarsening}
    For every $i \in \{1,\ldots, n\}$ we have that $\PP_{i}$ is a coarsening of $\PP_{i-1}$ . 
\end{lemma}

\begin{proof}
    Let $P \in \PP_{i-1}$ be arbitrary. We will argue that either $P \in \F_{i}$ or $P$ is a subset of some $P' \in \B_i$. If $P \in \F_{i-1}$, then clearly  $P \in \F_{i}$. If $P \in \B_{i-1}$ and $a_i$ is the only neighbor of $P$ in $G(A_{i-1},\B_{i-1})$, then again $P \in \F_{i}$. Finally, if $P \in \B_{i-1}$ and the neighborhood of $P$ in $G(A_{i-1},\B_{i-1})$ is a non-empty set $S$ different from $\{a_i\}$, then let $P'\coloneqq \{v \in B\ |\ N_G(v)\cap A_i = S \setminus \{a_i\}\}$. Then $P \subseteq P'$ and $P' \in B_i$, as required.
\end{proof}

We will also use the following lemma.
\begin{lemma}\label{lem:coarsening_twins}
    Let $G=(A,B,E)$ be a bipartite graph and let $\PP$, $\PP'$ be partitions of $B$ such that $\PP'$ is a coarsening of~$\PP$. Assume that $u,v \in A$ are $k$-near-twins in $G(A,\PP)$. Then $u,v$ are $k$-near-twins in $G(A,\PP')$. 
\end{lemma}
\begin{proof}
    Let $P$ be any part of $\PP'$ such that the adjacency of $u$ and $v$ to $P$ differs in $G(A,\PP')$. Then there must exist a part $Q \in \PP$ such that $Q\subseteq P$ such that the adjacency of $u$ and $v$ to $Q$ differs in $G(A,\PP)$. Since there are at most $k$ such parts in $\PP$, there can be at most $k$ distinct parts $P \in \PP'$ as above.
\end{proof}

\begin{proof}[Proof of \Cref{lem:near_twins}]
Fix any $i \in \{1,\ldots, n-1\}$. We know that $a_i = u$, where $u,v$ are $\kHaussler$-near-twins of $G(A_{i-1},\B_{i-1})$ found in the $i$-th iteration of the algorithm. In this iteration $a_i$ was removed from the set $A_{i-1}$ to obtain $A_i$, but $v$ remained in the set $A_i$. Consequently, we know that  $v$ was removed from $G$ in a later iteration, that is, as $a_j$ for some $j > i$. 

Since $a_i, a_j$ are $\kHaussler$-near-twins of $G(A_{i-1},\B_{i-1})$,
and \(a_i,a_j\) are both non-adjacent to all the parts from $\F_{i-1}$ (as these parts consist of isolated vertices in $G[A_{i-1},B]$),
the two vertices are also $\kHaussler$-near-twins in $G(A_{i-1},\PP_{i-1})$.
Then $a_i, a_j$ are also $\kHaussler$-near-twins in $G(A,\PP_{i-1})$, since replacing $A_{i-1}$ with $A$ does not affect the adjacency of $a_i, a_j$ to the right side of the graph. By combining and repeatedly applying Lemmas~\ref{lem:coarsening} and~\ref{lem:coarsening_twins} we then get that $a_i$ and $a_j$ are $\kHaussler$-near-twins also in $G(A, \PP_{i}), G(A, \PP_{i+1}),\ldots, G(A, \PP_{n})$. Since $\PP_n = \B_n \cup \F_n$ and $\B_n = \emptyset$, we have that $\PP_n = \F_n$. Then, since $\F_n = \F$, the result follows. 
\end{proof}

\subsubsection{Proof of \Cref{lem:quotient_no_halfgraph}}
\label{subsec:no_halfgraph}

Before proving \Cref{lem:quotient_no_halfgraph},
we note that induced $1$-subdivisions of complete graphs are obstructions both for monadic stability and for linear neighborhood complexity.

\begin{lemma}\label{lem:nosubdivclique}
	Let $\Cc$ be a graph class such for every $t\in \N$, the $1$-subdivision of $K_t$ is an induced subgraph of a graph from $\Cc$. Then $\C$ is not monadically stable.

\end{lemma}

\begin{proof}
    To see that $\C$ is not monadically stable, note that there is a fixed copyless transduction that from a $1$-subdivision $G$ of $K_{2t}$ transduces the half-graph $H_t$ as follows: Consider the $2t$ principal vertices of~$G$, denote them $u_1,\ldots, u_t, v_1,\ldots, v_t$ and mark them with a unary predicate $Q$.  Then mark, for every $i \le j$, the subdivision vertex on the path between $u_i$ and $v_j$ by a unary predicate $P$. Call the resulting colored graph $G^*$. The formula $\psi$ of the transduction then makes an edge between any two vertices $x$ and $y$ if they are both marked with $Q$ and they have a common neighbor that is marked with $P$. This creates edges precisely between vertices $u_iv_j$ with $i\le j$. The resulting graph is then the subgraph of $\psi(G^*)$ induced by the vertices that were marked with predicate $Q$.

\end{proof}

In the proof of \Cref{lem:quotient_no_halfgraph} we will need the following auxiliary Ramsey-like lemma.

\begin{lemma}\label{lem:auxiliary}
    There exists a function $\beta\colon \N \times \N \to \N$ with the following property. Let $G(A,B,E)$ be a bipartite graph with $A = \{\ell_1,\ldots, \ell_m\}$.  Let $\PP= \{P_1, \ldots, P_m\}$ be a partition of $B$ such that
    \begin{itemize}[nosep]
        \item for every $i\in [m]$ we have that $P_i \subseteq N_G(\ell_i)$, and
        \item for every $i\in [m]$ we have that $N_G(\ell_i) \cap P_j = \emptyset$ for all $j < i$ and $N_G(\ell_i) \cap P_j \not= \emptyset$ for all $j > i$.
    \end{itemize}
    Suppose further that $m>\beta(a,b)$ for some \(a,b \in \N\). Then $G$ contains $H_a$ (the half-graph of order \(a\)) as an induced subgraph or a $1$-subdivision of $K_b$ as an induced subgraph in which all principal vertices are in $A$. 
\end{lemma}
\begin{proof}
	For $a = 1$ or $b = 1$ we set $\beta(1,n) = \beta(n,1) = 1$. For $a>1$ and $b>1$ we set \[\beta(a,b) \coloneqq  (\beta(a-1,b) + \beta(a,b-1))^2 + 1.\] 
	We prove the statement by induction on $a+b$. 
	In the base case when $a+b = 2$ we have $a=b=1$, and the graph $G$ contains (by the assumption of the lemma with \(m=1\)) an edge from $\ell_1$ to a vertex in $P_1$, and thus contains $H_1$ as induced subgraph.
	
	For $a+b > 2$ we distinguish two possibilities:
	\begin{enumerate}
		\item There is a vertex $v \in P_m$ adjacent to at least $\sqrt{m-1}$ vertices among $\{\ell_1,\ldots,\ell_{m-1}\}$. 
		In this case, let $\{\ell_{i_1},\ldots,\ell_{i_t}\}$, be these neighbors of \(v\). 
		Note that $t \ge \sqrt{m-1} \ge \beta(a-1,b)$. 
		We set $A'\coloneqq  \{\ell_{i_1},\ldots,\ell_{i_t}\}$ and $B'\coloneqq  P_{i_1} \cup \ldots \cup P_{i_t}$. 
		The graph $G' = G[A',B']$ together with the partition $\PP'\coloneqq  \{ P_{i_1}, \ldots, P_{i_t} \}$ satisfies the assumption of the lemma for $\beta(a-1,b)$, and so $G'$ contains by induction either $H_{a-1}$ or a $1$-subdivision of $K_{b}$ as induced subgraph.
		In the latter case we are done, and in the former case the copy of $H_{a-1}$ in $G'$ together with $\ell_m$ and $v$ form an induced $H_{a}$ in $G$ (recall that by Remark~\ref{rem:ladders_sides} the copy of $H_{a-1}$ is embedded `nicely' in $G'$, and so $\ell_m$ and $v$ indeed extend $H_{a-1}$ to $H_a$).
		
		\item Every $v \in P_m$  is adjacent to fewer than $\sqrt{m-1}$ vertices among $\{\ell_1,\ldots,\ell_{m-1}\}$.
		In this case, let $M \subseteq P_m$ be a set of minimum size that dominates  $\{\ell_1,\ldots,\ell_{m-1}\}$, i.e., such that $\{\ell_1,\ldots,\ell_{m-1}\} \subseteq \bigcup_{v \in M} N_G(v)$. 
		Since every $v \in P_m$  has  fewer than $\sqrt{m-1}$ neighbors in $\{\ell_1,\ldots,\ell_{m-1}\}$, we have that $|M| \ge \sqrt{m-1} \ge \beta(a,b-1)$.

		Moreover, for every $v \in M$ there exists index $i(v) \in [m-1]$ such that $\ell_{i(v)}$ is adjacent to $v$ and not adjacent to any vertex in $M \setminus \{v\}$. 
		Thus, the set $M$ together with the set $N\coloneqq \{\ell_{i(v)} : v \in M\}$ form an induced matching of order at least $\beta(a,b-1)$ in $G$. 
		We now consider the subgraph of $G'$ of $G$ induced by $A'\coloneqq N$ and $B'\coloneqq  \bigcup_{v \in M} P_{i(v)}$. 
		The graph $G'$ together with the partition $\PP' = \{P_{i(v)} : v \in M\}$ of \(P'\) satisfies the assumptions of the lemma for $\beta(a,b-1)$, and so $G'$ contains by induction either $H_a$ or a $1$-subdivision of $K_{b-1}$ as an induced subgraph. 
		If $G'$ contains an induced $H_a$ we are done. 
		Otherwise $G'$ contains a $1$-subdivision of $K_{b-1}$ as an induced subgraph in which all the principal vertices are in $A'$. 
		Then, after adding $\ell_m$ together with $M$ we find a $1$-subdivision of $K_{b}$ as an induced subgraph of $G$ in which all the principal vertices are in~$A$.\qedhere
	\end{enumerate}
\end{proof}

\begin{proof}[Proof of \Cref{lem:quotient_no_halfgraph}]
Assume towards contradiction that for every $t$ there is some $G \in \C$ such that $H_t$, the half-graph of order \(t\), is an induced subgraph of $G(A,\F)$. Let $\{x_1,\ldots,x_t\}$ and $\{y_1,\ldots,y_t\}$ be the two sides of this induced $H_t$. For each $i \in [t]$, one of $x_i, y_i$ has to be in $\F$, because $x_iy_i \in E(H_t)$ and $G(A,\F)$ is bipartite. Consequently, there are $t$ parts in $\F$ that are vertices of the induced copy of $H_t$; let us order these parts ascendingly by the time of their freezing and denote them by $P_1,\ldots, P_t$, and let $\ell_1, \ldots, \ell_t$ be their leaders. Note that since $P_1,\ldots, P_t$ are ordered by the time of their freezing,  for any $i < j$ we have that $P_i \cap N_G(\ell_j) = \emptyset$,
because when the part $P_i$ was frozen,
say at time \(\tau\), it did not have any edges towards the vertices in $A_{\tau}$, and we know that $\ell_j \in A_{\tau}$ since the part $P_j$ was frozen after part $P_i$. 

We create an auxiliary graph $K$ with vertex set $[t]$ that contains edges between $i$ and $j$ with $i < j$ if and only if $N_G(\ell_i) \cap P_j\not=\emptyset$.
By Ramsey's theorem, graph $K$ contains a clique or independent set of size at least $\frac{1}{2}\log_2 t$.
We distinguish two cases:
\begin{enumerate}
    \item $K$ contains a clique of size $m \ge \frac{1}{2}\log_2 t$. Let $i_1,\ldots, i_m$ be the vertices of this clique. Then vertices $\ell_{i_1},\ldots, \ell_{i_m}$ and sets $P_{i_1},\ldots, P_{i_m}$ satisfy the assumptions of \Cref{lem:auxiliary}. 
        Therefore for any \(p\) with \(m > \beta(p,p)\), the graph $G$ contains either $H_p$ or a $1$-subdivision of $K_p$ as an induced subgraph.
        Since we assumed that $t$ can be made arbitrarily large, and since $m \ge \frac{1}{2}\log t$, we know that $m$ can be made arbitrarily large. Consequently, for any choice of $p$ we can obtain $m >  \beta(p,p)$, 
    and so we know that graphs from $\C$ contain arbitrarily large half-graphs or arbitrarily large 1-subdivisions of cliques as induced subgraphs. Thus either by definition or by \Cref{lem:nosubdivclique}, \(\C\) is not monadically stable, a contradiction.
    \item $K$ contains an independent set of size $m \ge \frac{1}{2}\log_2 t$. Let $i_1,\ldots, i_m$ be the vertices of this independent set. By the assumptions made at the beginning of this proof, with $\F'\coloneqq \{P_{i_1},\ldots, P_{i_m}\}$, the graph $G(A,\F')$ contains $H_m$ as an induced subgraph. Let $v_{1},\ldots, v_{m}$ be the vertices of this $H_m$ that are contained in $A$. We will now argue that there is a first-order formula $\psi(x,y)$ independent of $G$ (depending only on $\C$) and a marking of vertices of $G$ such that $H_m$ is an induced subgraph of $\psi(G)$. We start with the marking. We mark the leaders $\ell_{i_1}, \ldots, \ell_{i_m}$ with a unary predicate $U$ and we mark all vertices in $P_{i_1} \cup \ldots \cup P_{i_m}$ with a unary predicate $Q$.
	Within the subgraph of \(G\) induced on $P_{i_1} \cup \ldots \cup P_{i_m} \cup \{\ell_{i_1}, \ldots, \ell_{i_m}\}$, every part $P_{i_j}$ is precisely the neighborhood of $\ell_{i_j}$. 
    Therefore, the formula 
    $$\mathsf{samepart}(x,y)\coloneqq  Q(x) \land Q(y) \land \exists z \bigl(U(z) \land \adj(x,z) \land  \adj(y,z)\bigr).$$
    expresses that two vertices $u,v$ are in the same part $P_{i_j}$ for some $j$. 

	We now define $\psi(x,y)$ as the symmetric version of the formula $\exists z (\mathsf{samepart}(y,z) \land \adj(x,z))$. 
	This formula creates edges (in $\psi(G)$) between any $v \in A$ and all vertices of $P_{i_j}$ if there exists $u \in P_{i_j}$ such that $vu$ is an edge of $G$. Thus, in particular, vertex $v_1$ is adjacent to all the vertices in $P_{i_1} \cup \ldots \cup P_{i_m}$ in $\psi(G)$, vertex $v_2$ is adjacent to all the vertices in $P_{i_2} \cup \ldots \cup P_{i_m}$, and in general, vertex $v_j$ is adjacent to all the vertices from $P_{i_j} \cup \ldots \cup P_{i_m}$. This means that $\psi(G)$ contains a half-graph of order $m$ as an induced subgraph (for every $j \in [m]$ we can just pick one arbitrary vertex $u_{i_j}$ from $P_{i_j}$). Since we assumed that $t$ can be arbitrarily large and $m \ge \frac{1}{2}\log_2 t$, we know that graphs $\psi(G)$ contain arbitrarily large half-graphs as induced subgraphs. The marking of $G$ together with formula $\psi$ thus determine a (non-copying) transduction $\T$ of $\C$ such that $\T(\C)$ contains arbitrarily large half-graphs, a contradiction to $\C$ being monadically stable.\qedhere
 \end{enumerate}
\end{proof}

\subsection{Definability of $\F$}\label{sec:definability}

It remains to prove that the partition $\F$ is definable in the sense claimed by \Cref{lem:main}. To achieve this we provide a formula $\leader(x,y)$ and show how to define a suitable marking of vertices of every $G \in \C$.

We start with the marking of vertices of $G$.
Let $\F$ be the partition of $B$ computed by the algorithm and set $m\coloneqq |\F|$. Let $P_1,\ldots, P_{m}$ be the parts of $\F$ and let $\ell_1,\ldots, \ell_{m}$ be the corresponding leaders. Define an auxiliary graph $K$ with vertex set $[m]$ in which $ij$ is an edge if and only if $\ell_i$ has a neighbor in $P_j$ or $\ell_j$ has a neighbor in $P_i$. Since the vertices 
$\ell_1,\ldots, \ell_{m}$ are all from $A$, and since (by \Cref{lem:bd_degree}) every vertex in $A$ has a neighbor in at most $k^*$ parts from $\F$, the graph $K$ has at most $k^*m$ edges. Therefore, $K$ has a vertex of degree at most $2k^*$. Moreover, from the the definition of $K$ it follows that if $S$ is a subset of~$[m]$, then $K[S]$ has at most $k^*|S|$ edges and again has a vertex of degree at most $2k^*$. It follows that $K$ is $2k^*$-degenerate, and therefore there exists a proper coloring $\lambda$ of $[m]$ that uses at most $2k^*+1$ colors.
Let $C$ denote the set of colors used by this coloring.
We introduce $2k^*+1$ unary predicates \(M_c\), one for each color $c\in C$. For each $i \in [m]$ mark $\ell_i$ and all the vertices in $P_i$ with a unary predicate $M_c$, where $c= \lambda(i)$. This construction guarantees that whenever some $\ell_i$ is adjacent to some $v \in B$ satisfying the same predicate~$M_c$, then we have that $v \in P_i$. Clearly, the opposite is also true -- if $v\in P_i$, then it is adjacent to $\ell_i$ and it is marked with the same unary predicate.

We can then define
$$ \leader(x,y) \coloneqq  \bigvee_{c \in C} (L(x) \land M_c(x) \land M_c(y) \land \adj(x,y) ).$$
This formula first checks whether $x$ is in the left side of $G$, and then checks the condition outlined above.
This completes the proof of \Cref{lem:main}.

\subsection{Neighborhood covers}
\label{sec:mainlemma-implies-covers}

In this section we prove \Cref{lem:covers}. We recall the relevant definitions and the statement of the lemma for convenience.
A {\em{neighborhood cover}} of a graph $G$ is a collection \(\cal X\) of subsets of \(V(G)\) called \emph{clusters} such that for each \(v \in V(G)\) there is  \( X \in \cal X\) such that  \(N_G[v] \subseteq X \). We say that a neighborhood cover \(\cal X\) is \emph{compact} if there exists a partition \(\PP\) of \(V(G)\) such that (i) for each \(P \in \PP\) there is \(u \in V(G)\) with \(P \subseteq N_G[u]\) and (ii) for each \( X \in \cal X\) we have that \( X = N_G[P]\) for some \(P \in \PP\). Furthermore, we say that a neighborhood cover \(\cal X\) has \emph{overlap} \(k\) if every vertex \(v\) is contained in at most \(k\) clusters of \(\cal X\).

\lemmacovers*

\begin{proof}
For a graph $G$ we define $\double(G)$ to be the bipartite graph $(A,B,E)$ with sides $A = V(G)\times\{0\}$, $B = V(G)\times\{1\}$ and edge set
$ E = \{ (u,0)(v,1) \colon uv\in E(G) \} \cup \{(u,0)(u,1) \colon u \in V(G)\}$. 
Thus, $\double(G)$ is the standard ``bipartization'' of a graph in which we also add an edge between the two copies of every vertex. Moreover, let $\rho$ be the function that  projects the vertices from $\double(G)$ to the vertices from the original graph $G$. That is, we set $\rho((v,0)) =\rho((v,1))\coloneqq v$ for all $v \in V(G)$.
For a vertex set $S \subseteq V(\double(G))$ we define $\rho(S) \coloneqq \{\rho(v)\ \colon\ v \in S\}$.

It is easily seen that $\{\double(G)\ \colon\ G \in \C\}$ is transducible from $\C$, and hence is also monadically stable and has inherently linear neighborhood complexity. 
Consequently, we can apply the algorithm from \Cref{lem:main} to this class.
Let $k^* \in \N$ be the associated constant claimed in the lemma.

Let $G \in \C$. Note that $\double(G)$ does not have any isolated vertices by construction. We apply the algorithm from \Cref{lem:main} to $\double(G)$. This gives us a partition $\F$ of $B$ and vertices $\ell_1,\ldots, \ell_{|\F|}$ with the claimed properties.

 We then set ${\cal X}\coloneqq \{N_G[\rho(P)]\ \colon\ P \in \F\}$  and claim that this is a compact neighborhood cover of the original graph $G$ with overlap $k\coloneqq k^*$.
To see this, we set $\PP\coloneqq \{\rho(P)\ \colon\ P \in \F\}$. 
Clearly this is a partition of $V(G)$, and we have that each $X \in \cal X$ is of the form $N_G[Q]$, where $Q \in \PP$.

Since in $\double(G)$ we have that $P_i \subseteq N(\ell_i)$ for each $i \in \{1,\ldots, |\F|\}$, we know that each vertex in $\rho(P_i)$ is a neighbor of $\rho(\ell_i)$ in $G$ or is $\rho(\ell_i)$ itself. Consequently, we have that $\rho(P_i) \subseteq N_G[\rho(\ell_i)]$ for each $i$, which means that
the requirement (i) of $\cal X$ being compact is satisfied.
Requirement (ii) is satisfied by definition.
This proves that $\cal X$ is compact. Finally, it remains to show that each $v \in V(G)$ is in at most $k^*=k$ clusters of $\cal X$. Since each part $X$  of $\cal X$ is of the form $N_G[Q]$ for some $Q \in \PP$, we need to show that $v$ is in at most $k^*$ sets of this form. For each $Q$ with $v \in N_G[Q]$, we have that $v \in Q$ or $v \in N_G(Q)$, and we will count the number of clusters containing $v$ based on these two options.
First, since $\PP$ is a partition of $V(G)$, we know that there is exactly one part $Q$ of $\PP$ such that $v \in Q$; let us call this part $Q_v$. 
Second, from Lemma~\ref{lem:main} we know that $v'\coloneqq (v,0)$ has a neighbor in at most $k^*$ parts $P$ from $\F$. Therefore, we have that $v \in N_G(Q)$ for at most $k^*$ parts $Q=\rho(P)$ from $\PP$. By combining the options for $v \in Q$ and $ v \in N_G(Q)$ as above, we get that $v$ is contained in at most $k^*+1$ clusters of $\cal X$. However, one easily checks that the part $P \in \F$ with $Q_v=\rho(P)$  is one of the $k^*$ parts in which $(v,0)$ has a neighbor (we have that $(v,1) \in P$, and this is a neighbor of $(v,0)$), and therefore the cluster $N_G[Q_v]$ is already accounted for.
Consequently, $v$ is in at most $k^*$ clusters of $\cal X$.
\end{proof}

\section{Efficient sparsification}
\label{sec:sparsify}

In this section we prove \Cref{thm:main}, restated below for convenience.

\thmmain*

The remainder of this section is devoted to the proof of \Cref{thm:main}.
First, in \Cref{sec:bipartite} we show that it is sufficient to prove \Cref{thm:main} in the setting of bipartite graphs.
In \Cref{sec:branchings}, we introduce a useful tool that allows us to recursively decompose the input graph. Then in \Cref{sec:algo_sparsify}, we describe the algorithm $\Aa$ that for each $G \in \C$ computes a graph $H$ as claimed in \Cref{thm:main}. We will call this graph $S(G)$, and will define the class $\D\coloneqq  \{S(G)\ :\ G \in \C\}$.
The next sections are devoted to analyzing the properties of graphs $S(G)$ (\Cref{sec:height}), bounding the runtime of $\Aa$ (\Cref{sec:runtime}), and defining a suitable interpretation~$I$ (\Cref{sec:interp}). After this, it remains to prove that the class $\D$ defined in \Cref{sec:algo_sparsify} is indeed of bounded expansion. This is technically the most challenging part of the proof of \Cref{thm:main} and is split into showing that $\D$ is transducible from $\C$ (\Cref{sec:transduce}), and then using this to finally prove the result in \Cref{sec:be}.

\subsection{Bipartite Graphs}\label{sec:bipartite}

It will be more convenient for us to work with bipartite graphs.
To see that focusing on bipartite graphs does not restrict the scope of our results, we now argue 
that \Cref{thm:main} with the additional assumption that \(\C\) is a class of bipartite graphs implies the general version of \Cref{thm:main}.

Let $\C$ thus be a class of (not necessarily bipartite) graphs satisfying the assumptions of \Cref{thm:main}. For every $G \in \C$ we define a graph $B(G)$ as follows. The vertex set of $B(G)$ is $V(G) \times \{L,R,P_1,P_2\}$. The edge set of $B(G)$ consists of the following edges: For every pair of vertices $u,v$ with $uv \in E(G)$ there are edges $(u,L)(v,R)$ and $(v,L)(u,R)$ in \(B(G)\). Moreover, for each $u \in V(G)$ there are edges in $B(G)$ forming the path $(u,L)(u,P_1)(u,P_2)(u,R)$. It is easy to check that $B(G)$ is a bipartite graph -- one side is formed by vertices from $V(G)\times \{L,P_2\}$ and the other side by vertices from $V(G)\times \{R,P_1\}$. We then define $\C'\coloneqq \{B(G)~:~G \in \C\}$. 
We argue that there is an interpretation $J$ such that for any $G$ there exists a marking of $B(G)$ by unary predicates such that $G = J(B(G))$.
The interpretation first identifies the four copies of each vertex using \(L,R,P_1\)- and \(P_2\)-predicates,
then connects two vertices \((u,L)\) and \((v,L)\) if \((u,L)\) is adjacent to \((v,R)\), and finally only outputs the \(L\)-vertices.
Also note that $\C'$ is transducible from $\C$ and that for each $G \in \C$ we can easily compute $B(G)$ in quadratic time. Since $\C'$ is transducible from~$\C$, it also satisfies the assumptions of \Cref{thm:main}. Therefore, we can apply \Cref{thm:main} to $\C'$ and obtain an algorithm $\Aa'$, interpretation $I'$ and class of graphs $\D'$ of bounded expansion with the properties claimed there. 
Let us argue that \Cref{thm:main} also holds for the graph class \(\C\),
by specifying the required
algorithm $\Aa$, interpretation $I$ and class of graphs $\D$ as follows.
We set $\D\coloneqq \D'$. Algorithm $\Aa$ takes as input graph $G \in \C$, computes $B(G)$, applies algorithm $\Aa'$ to $B(G)$, and returns its output, $H$. The interpretation $I$ works by composing $J$ with $I'$, that is, it first applies $I'$ to obtain $B(G)$ from $H$, and then applies $J$ to $B(G)$ to obtain $G$. The markings of $H$ and $B(G)$ used by $I'$ and $J$ can be combined in such a way that we indeed have $G = (J \circ I)(H)$.

This shows that it is sufficient to prove \Cref{thm:main} for the special case of bipartite graphs.
For the rest of the section we fix a monadically stable class $\C$ of bipartite graphs as in the assumptions of \Cref{thm:main}. It will be convenient for us to assume that $\C$ is closed under induced subgraphs and bipartite complements. This comes at no cost of generality, since the properties of being monadically stable and of having inherently linear neighborhood complexity remain preserved if we close $\C$ under taking induced subgraphs and bipartite complements.

\subsection{Branchings}\label{sec:branchings}

We now give the precise definition of \emph{branchings} outlined in \Cref{sec:overview},
and derive the existence of branchings with small overlap from
\Cref{lem:main}.
Such branchings are a crucial ingredient for constructing $S(G)$.

\begin{definition}\label{def:branching}
    Let $G \in \C$. A \emph{left branching} of $G$ is a collection $G_1,\ldots, G_m$ of induced subgraphs of $G$ satisfying the following conditions:
    \begin{enumerate}[nosep]
        \item There are distinct vertices $\ell_1,\ldots, \ell_m$ in $L(G)$ such that for each $i \in [m]$ we have that $\ell_i \in L(G_i)$ and $\ell_i$ dominates $R(G_i)$ in $G_i$ (that is, $R(G_i) \subseteq N_{G_i}(\ell_i)$).
        \item For every $uv \in E(G)$ there exists an $i \in [m]$ such that $u,v \in V(G_i)$ 

        \item The sets $R(G_1),\ldots, R(G_m)$ form a partition of $R(G)$.
    \end{enumerate}
    The graphs $G_1,\ldots,G_m$ are called the \emph{branches} and the vertices $\ell_1,\ldots, \ell_m$ are their \emph{leaders}.
    If every vertex of $G$ belongs to at most $k$ graphs from $G_1,\ldots, G_m$, then we say that the branching has \emph{overlap $k$}.
    A \emph{right branching} of $G$ is defined analogously with the roles of $R(G)$ and $L(G)$ reversed.
\end{definition}

\begin{figure}[t]
  \begin{minipage}[t]{0.30\textwidth}
    \vspace*{0pt}
    \begin{center}
    \includegraphics[scale=1.1]{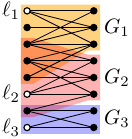}
    \end{center}
  \end{minipage}\hfill
  \begin{minipage}[t]{0.70\textwidth}
    \vspace*{5pt}
    \captionof{figure}{
        A left branching of a bipartite graph into three branches \(G_1,G_2,G_3\) 
        with respective leaders \(\ell_1,\ell_2,\ell_3\).
        As required, each edge is in at least one branch, and each right vertex is in exactly one branch.
        Each left vertex is in at most two branches, and thus the branching has overlap~two.
    }
    \label{fig:branching}
  \end{minipage}
\end{figure}

\Cref{fig:branching} illustrates \Cref{def:branching}.
The following lemma is a simple consequence of \Cref{lem:main}.
\begin{lemma}\label{lem:branching_compute}
    There exists $k \in \N$ and an algorithm with runtime $\Oh(|G|^4)$ that computes for any $G \in \C$ a left (or right)
    branching of overlap at most $k$.

    Further, there is a first-order formula $\branch(x,y)$
    such that for any $G \in \C$ and branching $G_1,\ldots, G_m$ computed by the algorithm, there is a unary coloring \(G^+\) of \(G\)
    such that for any $u,v \in V(G)$ we have that $G^+ \models \branch(u,v)$ if and only if $u$ is the leader of $G_i$ for some $i \in [m]$
    and $v \in V(G_i)$.
\end{lemma}
\begin{proof}
    We focus on left branchings, as the case of right branchings is analogous. Let $k^*$ be the constant provided by \Cref{lem:main}; we shall prove the statement for $k\coloneqq k^*$. We apply the algorithm from \Cref{lem:main} to $G$ and obtain sets $P_1,\ldots, P_m \subseteq R(G)$ and vertices $\ell_1,\ldots, \ell_m$. For each $i \in [m]$ we set $G_i\coloneqq  G[N(P_i), P_i]$. The first and third properties from the definition of branchings are satisfied by \Cref{lem:main}.
    The second property holds because $P_1,\ldots,P_m$ form a partition of $R(G)$, and the neighborhoods of these sets lie within $G_1,\ldots,G_m$. Hence each edge of $G$ is contained in at least one of these graphs.
    Finally, every vertex from $R(G)$ is in exactly one $G_i$. By the degree bound of \Cref{lem:main}, every vertex in $L(G)$ is in at most $k$ graphs from \(G_1,\dots,G_m\). Hence, the overlap of the branching is at most~$k$.
\end{proof}

The following simple observation will be useful in the next subsection. 
\begin{lemma}\label{lem:isolated}
Let $G \in \C$. 
If $G_1,\ldots, G_m$ is a left (resp. right) branching of $G$, then for every $i \in [m]$ the bipartite complement $\widetilde{G}_i$ of $G_i$ contains an isolated vertex in $L(\widetilde{G}_i)$ (resp. in $R(\widetilde{G}_i)$).

\end{lemma}
\begin{proof}
    The leader of $G_i$ becomes an isolated vertex in the complement, as it is adjacent to all the vertices from the other side of \(G_i\).

\end{proof}

\subsection{The algorithm and the definition of $\D$}
\label{sec:algo_sparsify}

For any input graph $G \in \C$, the construction of $S(G)$ proceeds in two phases. 
We first compute a tree $T(G)$ that encodes the properties of the edge relation of $G$. Then we add additional edges between $V(G)$ and $T(G)$ to obtain $S(G)$. 

We now describe the algorithm that computes the tree $T(G)$. 
To each node $p$ of $T(G)$ we associate a graph $G(p)$ that is either an induced subgraph of $G$ or the bipartite complement of an induced subgraph of~$G$.
We start by creating the root vertex $r$ and set $G(r) \coloneqq  G$. We then recursively create the rest of $T(G)$ by applying the following extension rules to each previously added node \(p\).
\begin{enumerate}
    \item If $G(p)$ is not edgeless and is at odd depth, we remove all its isolated vertices (if it has any) and compute its left branching $G_1,\ldots, G_m$. For each branch $G_i$ we create a child $p_i$ and set $G(p_i)$ to be the bipartite complement of $G_i$.
    \item If $G(p)$ is not edgeless and is at even depth, we remove all its isolated vertices (if it has any) and compute its right branching $G_1,\ldots, G_m$. For each branch $G_i$ we create a child $p_i$ and set $G(p_i)$ to be the bipartite complement of $G_i$.
    \item If $G(p)$ is edgeless, then $p$ has no children. 
\end{enumerate}
We note that by \Cref{lem:isolated}, every internal node of $T(G)$, possibly with the exception of the root, contains an isolated vertex. Therefore, in each recursive step the graph we work with gets smaller, and so the recursion ends after finitely many steps. 

We construct the graph $S(G)$ as follows. The vertex set of $S(G)$ is $V(T(G)) \cup V(G)$. 

The vertices of \(T(G)\) are called \emph{nodes}, and we will use the letters $p,q,r,s$ to denote them.
The edge set of $S(G)$ consists of the edges of $T(G)$ together with all edges $vp$ such that $p \in V(T(G))$ and $v \in V(G(p)) \subseteq V(G)$. Finally, we mark the root of $T(G)$ by a unary predicate $U$, vertices from $L(G)$ by a unary predicate $L$ and vertices from $R(G)$ by a unary predicate $R$. 
We define the \emph{height} of $S(G)$ to be the height of $T(G)$. 
Finally, we define $\D\coloneqq  \{S(G)\ \colon\ G \in \C\}$.

\subsection{$T(G)$ has bounded height}
\label{sec:height}

\begin{lemma}\label{lem:half-graph_grow}
    Let $G\in \C$ and let $T(G)$ be the tree computed by the algorithm. Let $p$ and $p''$ be nodes of $T(G)$ such that $p$ is the grandparent of $p''$. Suppose $G(p'')$ contains an induced half-graph of order $t$. Then $G(p)$ contains an induced half-graph of order $t+1$. 
\end{lemma}
\begin{proof}
    Let $v_1,\ldots, v_t \in L(G(p''))$ and $u_1,\ldots, u_t \in R(G(p''))$ be the vertices of a half-graph of order $t$ in $G(p'')$ (recall that we may assume that the half-graph is contained in $G(p'')$ in this fashion by Remark~\ref{rem:ladders_sides}). We first note that since $G(p'')$ is an induced subgraph of $G(p)$, the vertices also induce a half-graph of order $t$ in $G(p)$. We now show how to increase the order by one by 
    finding vertices $v_{t+1} \in L(G(p))$ and $u_{t+1} \in R(G(p))$ such that $v_{t+1}$ is adjacent to $u_1,\ldots, u_{t+1}$ and $u_{t+1}$ is non-adjacent to $v_1,\ldots, v_t$.

	We will assume that $p$ is at even depth; the argument for odd depth is analogous.
    Let $G_1,\ldots, G_{m_1}$ be the (left) branching of $G(p)$ (after removing isolated vertices). Let $p'$ be the child of $p$ that is the parent of $p''$. Let $i \in [m_1]$ be such that $G(p')$ is the bipartite complement of $G_i$, and let $\ell_i$ be the leader of $G_i$. Note that  $\ell_i$ is adjacent to all the vertices of $R(G_i)$ in graph $G_i$, and is isolated in $G(p')$. Let $G_1',\ldots, G_{m_2}'$ be the branching of $G(p')$ (after removing isolated vertices) and let $G_j'$ be the branch such that $G(p'')$ is the bipartite complement of $G_j'$. Let $\ell_j'$ be the leader of $G_j'$. Then $\ell_j'$ is isolated in $G(p'')$. We claim that $v_{t+1}\coloneqq \ell_i$ and $u_{t+1}\coloneqq \ell_j'$ have the desired properties. Clearly, since $\ell_j'$ is isolated in $G(p'')$ and since $G(p'')$ is an induced subgraph of $G(p)$, vertex $u_{t+1}$ is non-adjacent to  vertices $v_1,\ldots, v_t$. Also, since every vertex from $u_1,\ldots, u_t$ is adjacent to at least one vertex in $v_1,\ldots, v_t$, we know that $u_{t+1}$ has to be distinct from all the vertices $u_1,\ldots, u_t$. For $v_{t+1} = \ell_i$, we know that it is adjacent to all the vertices in $R(G(p'))$. Since $u_1,\ldots,u_{t+1} \in R(G(p'))$, we therefore know that $v_{t+1}$ is adjacent to all these vertices. Also, since $\ell_i$ is isolated in $G(p')$, it is removed from it before branching, and therefore it is not a vertex of $G(p'')$. Consequently, $v_{t+1}$ is distinct from all vertices $v_1,\ldots, v_t$.
\end{proof}

\begin{lemma}\label{lem:tree-height}
    There exists $h \in \N$ depending only on $\C$ such that $T(G)$ has height at most $h$, for every $G \in \C$.
\end{lemma}
\begin{proof}
    Since $\C$ is monadically stable, there exists $t\in \N$ such that no graph $G\in \C$ contains an induced half-graph of order more than $t$.
    Assume for contradiction that for some $G \in \C$, the tree $T(G)$ has height more than $2t$. Then there exists a non-leaf node $p$ of $T(G)$ of depth $2t$ in $T(G)$. Since $p$ is not a leaf, the graph $G(p)$ contains an edge, and therefore contains a half-graph of order $1$. By applying \Cref{lem:half-graph_grow} $t$ times, we obtain a half-graph of order $t+1$ in $G(r) = G$, a contradiction.
\end{proof}

\subsection{Running time analysis}
\label{sec:runtime}

\begin{lemma}\label{lem:runtime}
For every $G \in \C$ the computation of $T(G)$ and of $S(G)$ runs in time  $\mathcal{O}(|G|^4)$.
\end{lemma}
\begin{proof}
Let $h$ denote the height of $T(G)$ computed by the algorithm. We will show that the runtime of the algorithm is $\mathcal{O}_h(|G|^4)$,
where the subscript \(h\) in the \(\mathcal{O}\)-notation hides constant factors depending on \(h\).
By \Cref{lem:tree-height}, the height \(h\) is bounded by a constant and the statement of the lemma follows.

If $h=0$, then the algorithm only has to check whether $G$ is edgeless, and this clearly can be done in time $\mathcal{O}(\|G\|)\le\mathcal{O}(|G|^4)$.
If $h > 0$, then the algorithm removes all isolated vertices, computes the branching of $G$, and then recurses into (the bipartite complements of) each branch.
By \Cref{lem:branching_compute}, each vertex of $G$ is in at most $k$ of the branches \(G_1,\dots,G_m\), and thus \(\sum_{i=1}^m |G_i| \le k|G|\).
The computation of the branching can be done in time $\mathcal{O}(|G|^4)$, by \Cref{lem:branching_compute}.
Thus, by induction we can bound the runtime~by
\[
   \mathcal{O}(|G|^4) +
    \sum_{i=1}^m
   \mathcal{O}_{h-1}(|G_i|^4)
   \le 
   \mathcal{O}(|G|^4) +
   \mathcal{O}_{h-1}\biggl(\sum_{i=1}^m|G_i|\biggr)^4
   \le 
   \mathcal{O}(|G|^4) +
   \mathcal{O}_{h-1}(k|G|)^4
   \le 
   \mathcal{O}_{h}(|G|^4). \qedhere
\]
\end{proof}

\subsection{Interpretation $I$}
\label{sec:interp}

\begin{lemma}\label{lem:recover_G_from_S}
    There exists an interpretation $I$ such that $G = I(S(G))$ for every $G \in \C$.
\end{lemma}
\begin{proof}
    First we prove the following claim.
    Let $u \in L(G)$ and $v \in R(G)$ be two vertices and let $p$ be a node of $T(G)$ such that $u,v \in V(G(p))$ 
    and no child $q$ of $p$ satisfies $u,v \in V(G(q))$. 
    Then $uv \in E(G)$ if and only if the number $d$ of vertices on the path from $p$ to the root of $T(G)$ is even. 

    We prove this statement by induction on $d$.
    If $d=1$, then $p$ is the root of $T(G)$. If $p$ has no children, then $G$ is edgeless and we are done. Otherwise $p$ has children, and these children were created by a left branching $G_1,\ldots, G_m$. Assume for contradiction that $uv \in E(G)$. Then by \Cref{def:branching} there exists $i \in [m]$ such that $uv \in E(G_i)$, and so both $u$ and $v$ are in $V(G(p_i))$, a contradiction to the assumption that $p$ has no child with this property.

    For the induction step assume that $d>1$. Let $q$ be the child of the root that is on the path from the root to $p$. Note that $G(q)$ is the bipartite complement of an induced subgraph of $G$ that contains both $u$ and $v$.
    Moreover, the subtree of $T(G)$ rooted at $p$ satisfies the assumption of our claim for $d'=d-1$. Since the adjacency between $u$ and $v$ switches between the root graph and $G(q)$, and the parity of $d'$ is the opposite of the parity of $d$, the result follows. 

    We now describe an interpretation $I$ that reconstructs $G$ from $S(G)$. The formula $\psi(x,y)$ that determines the adjacency between $x$ and $y$ works as follows: It checks that $x$ and $y$ are vertices from \(G\) for which $L(x)$ and $R(y)$ holds, and finds a node $z$ of $T(G)$ that is adjacent to both $x$ and $y$ in \(S(G)\) such that no child of $z$ is adjacent to both $x$ and $y$ in \(S(G)\). It then counts the number of vertices on the path in \(T(G)\) from $z$ to the root. Since \Cref{lem:tree-height} bounds the height of $T(G)$ by a constant,
    and the root as well as the two sides of \(G\) are marked by separate unary predicates in \(S(G)\),
    all these conditions can be easily expressed in first-order logic. 
\end{proof}

\subsection{$\D$ is transducible from $\C$ }
\label{sec:transduce}

\begin{lemma}\label{lem:transducible}
    There is a transduction $\ST$ such that $S(G) \in \ST(G)$ for every \(G \in \C\).
\end{lemma}

In the proof of \Cref{lem:transducible} we will use the following lemma that tells us that the intersection graph of graphs in a branching computed by our algorithm is sparse.

\begin{lemma}\label{lem:branching_graph_sparse}
    There exists $d \in \N$ with the following property. For any $G \in \C$, let $G_1,\ldots, G_m$ be the branching of $G$ computed by the algorithm and let $\ell_1,\ldots, \ell_m$ be the leaders of the branches. Define an intersection graph $K(G)$ with vertex set $\{\ell_1,\ldots, \ell_m\}$ in which there is an edge between $\ell_i, \ell_j$ with $i\not=j$ if and only if $V(G_i) \cap V(G_j) \not=\emptyset$. Then \(K(G)\) is \(d\)-degenerate, and in particular admits a proper $(d+1)$-coloring.
\end{lemma}
\begin{proof}
    We proceed by assuming that the class $\{K(G)\ \colon\ G \in \C\}$ is not $d$-degenerate for any $d$,
    and show that this contradicts our assumption that \(\C\) has inherently linear neighborhood complexity.
    
    Fix an arbitrary \(d \in \N\).
    Towards contradiction, choose $G \in \C$ such that $K(G)$ defined as in the statement of the lemma is not $d$-degenerate.
    Let $k$ be the number from \Cref{lem:branching_compute} that bounds the overlap of the considered branching $G_1,\ldots, G_m$.
    By symmetry, we will assume that the branching is a left branching.
    To simplify the notation we set $A\coloneqq L(G)$ and $B\coloneqq R(G)$.
    The leaders $\ell_1,\ldots, \ell_m$ of the branching are contained in $A$.

    Since $K(G)$ is not $d$-degenerate, we know that its vertex set contains a subset $S$ such that in $K(G)[S]$ every vertex has degree larger than $d$. Consequently, the number of edges in $K(G)[S]$ is larger than $d|S|/2$.  
    For every edge $\ell_i\ell_j$ of $K(G)[S]$ and a vertex $v \in A$, we say that $v$ is \emph{responsible} for $\ell_i\ell_j$ if $v \in V(G_i) \cap V(G_j)$. Note that by the definition of edges of $K(G)$, we have that for every edge of $K(G)[S]$ there is at least one vertex in $A$ that is responsible for it (every vertex $v$ with $v \in V(G_i) \cap V(G_j)$ has to be in $A$, because every vertex in $B$ is in exactly one $G_i$ by \Cref{def:branching}). 
    We note that for each edge there can be more than one vertex responsible for it, and each vertex can be responsible for more than one edge. However, since each $v \in A$ is in at most $k$ branches, each vertex is responsible for at most ${k \choose 2}$ edges. 
    Let $P$ be an inclusion-wise minimal set of vertices such that for every edge $\ell_i\ell_j$ of $K(G)[S]$ there is some vertex $v \in P$ that is responsible for $\ell_i\ell_j$. 
    Since there are at least $d|S|/2$ edges in $K(G)[S]$ and each vertex is responsible for at most ${k \choose 2}$ edges, we see that 
    $|P| \ge d|S|/\bigl(2{k \choose 2}\bigr)$. 

    Let $H$ be the graph with vertex set $A$ and edges constructed as follows: For each $i \in [m]$, 
    we put an edge between $\ell_i$ and every $v \not=\ell_i$ with $v\in  V(G_i)$.
    Since the branching is definable as described in \cref{lem:main}, there is a transduction $\T$ (depending only on $\C$) such that $H \in \T(G)$. 
    Note that no two vertices $u,v  \in P \setminus S$ can have the same neighborhood in $S$ (in graph $H$), as this would contradict the minimality of $P$.

    Hence, in $H$ each of the vertices in  $P\setminus S$ is adjacent to a different neighborhood in $S$.
    Since there are at least 
 $d|S|/\bigl(2{k \choose 2}\bigr) - |S| = (d/(2{k \choose 2})-1)|S|$ vertices in $P \setminus S$ and since the constant \(d\) can be chosen arbitrarily large, and \(H \in \T(\C)\), this contradicts our assumption that
    \(\C\) has inherently linear neighborhood complexity.
\end{proof}

\begin{lemma}\label{lem:tree_coloring}
    There exists a finite set $C$ of colors such that every $G \in \C$ admits a coloring $\lambda: V(T(G)) \to C$ where all the nodes $p,q$ of the same color satisfy $V(G(p)) \cap V(G(q)) = \emptyset$.
\end{lemma}
\begin{proof}
    We set $C_0\coloneqq \{0\}$ and $C_{i+1}\coloneqq C_i \times [d+1]$ where $d$ is the constant from \Cref{lem:branching_graph_sparse}. 
    We set $C\coloneqq  C_0 \cup \ldots \cup C_h$, where $h$ is the constant from \Cref{lem:tree-height}.

    We color the root of $T(G)$ by color $0$ and proceed downwards coloring the remaining vertices. Let $p$ be any node that is already colored whose children are not yet colored, and let $c$ be the color of $p$. Let $G_1,\ldots, G_m$ be the branching of $G(p)$ computed by the algorithm. Let  $\alpha: \{\ell_1,\ldots,\ell_m\} \to [d+1]$ be a proper coloring of $K(G(p))$ from \Cref{lem:branching_graph_sparse}.  For each $i \in [m]$ we color the child $p_i$ of $p$ corresponding to (the bipartite complement of) $G_i$ with the color $\lambda(p_i)\coloneqq (c,\alpha(\ell_i)) \in C_i$.
    By induction on the depth $\delta$, one easily verifies that vertices of distance at most $\delta$ from the root are colored as claimed by the lemma.
\end{proof}

    The following lemma lets us relativize formulas to each color class given by \Cref{lem:tree_coloring}.

    \begin{lemma}\label{lem:relativize}
    For every pair of formulas \(\nu(x,y)\) and \(\phi(y,z)\) there is a formula \(\psi(x,y,z)\) such that the following holds.
    Consider a graph \(S\) and a vertex set \(P \subseteq V(S)\) where each \(p \in P\) is associated with an induced subgraph \(G(p)\) of \(S\) such that
    \begin{itemize}[nosep]
        \item the graphs \(G(p)\) with \(p \in P\) are pairwise vertex-disjoint, and
        \item for all \(p \in P\) and \(v \in V(S)\) we have \(v \in V(G(p))\Leftrightarrow S \models \nu(p,v)\).
    \end{itemize}
    For each $p \in P$, let $G^+(p)$ be obtained from $G(p)$ by adding some colors and 
    let \(S^+\) be the graph obtained by adding the colors of each disjoint subgraph \(G^+(p)\) onto \(S\).
    Then for all \(p \in P\) and \(u,v \in V(S)\),
    \[
       [u,v \in V(G^+(p)) \text{~and~} G^+(p) \models \phi(u,v)] \qquad\textrm{if and only if}\qquad S^+ \models \psi(p,u,v).
    \]
    \end{lemma}
    \begin{proof}
        We choose \(\psi(x,y,z) \coloneqq  \nu(x,y) \land \nu(x,z) \land \phi'(x,y,z)\),
        where \(\phi'\) is obtained from \(\phi\) by relativizing all quantifiers via \(\nu\).
        More precisely, if $\exists a\ \xi(a)$ is a subformula of $\phi(y,z)$ then this is replaced in $\phi'(x,y,z)$ by 
        $\exists a\ \nu(x,a) \land \xi(a)$, and similarly for universal quantifiers.
    \end{proof}

\begin{proof}[Proof of \Cref{lem:transducible}]
     For $0 \le d \le h$ we define $S_d(G)$ to be the graph obtained from $S(G)$ by
     marking all nodes at depth \(d\) with a unary predicate \(T_d\),
     deleting all nodes of $T(G)$ at depth larger than $d$, 
     marking vertices of \(G\) by another unary predicate, 
     and adding all the edges of $G$. 
     (The last two steps are well-defined since \(V(G) \subseteq V(S(G))\).)
     Let \(h\) be the height bound of \Cref{lem:tree-height}.
     Then in particular, $S_0(G)$ consists of $G$ together with a single vertex adjacent to all the vertices of $G$,
     and \(S_h(G)\) is the graph obtained from \(S(G)\) by adding the edges of \(G\).

     The graph \(S_0(G)\) can clearly be obtained from \(G\) by a fixed transduction \(\ST_0\).
     We claim that for each $1 \le d \le h$ there is a transduction $\ST_{d}$ that transduces $S_{d}(G)$ from $S_{d-1}(G)$. 
     Once we prove this, the lemma follows by composing these transductions, that is, by setting $\ST = \ET\circ \ST_h \circ \ldots \circ \ST_0$, where $\ET$ is the transduction that removes all edges from the copy of $G$ in $S_h(G)$.

    We fix \(0 \le d < h\) and describe a transduction that constructs \(S_{d+1}(G)\) from \(S_{d}(G)\).
    By the definition of \(S_d(G)\), 
    it is trivial to define a formula \(\nu(x,y)\) such that for all \(p,v \in V(S_d(G))\) we have 
    \(S_d(G) \models \nu(p,v)\) if and only if
    \(p \in T_d\) and \(v \in V(G(p))\).
    For each \(p \in T_d\), \Cref{lem:branching_compute} yields a coloring \(G^+(p)\) of \(G(p)\)
    and a branching such that for all $u,v \in V(G(p))$ we have $G^+(p) \models \branch(u,v)$ if and only if
    $u$ is the leader of a branch that contains \(v\).
    
    Let $\lambda : T_d \to C$ be the coloring of \(T(G)\) from \Cref{lem:tree_coloring}, restricted to \(T_d\).
    For each color \(c \in C\), we apply
    \Cref{lem:relativize} with
    \(\phi(y,z) \coloneqq  \branch(y,z)\), \(S = S_d(G)\), \(P \coloneqq  \{p \in T_d\ \colon\ \lambda(p) = c \}\),
    and \(G^+(p)\) being the colorings of the graphs \(G(p)\) defined above for each \(p \in P\).
    By \Cref{lem:tree_coloring}, these graphs are pairwise vertex-disjoint.
    Hence, we obtain
    a formula \(\psi^c(x,y,z)\) and a unary coloring \(S^c\) of \(S_d(G)\) such that for each \(p \in T_d\) of color \(c\) and \(u,v \in V(S_d(G))\),
    we have \(S^c \models \psi^c(p,u,v)\) if and only if \(u,v \in V(G(p))\) and $G^+(p) \models \branch(u,v)$.

    By aggregating these formulas and colorings over all \(c \in C\), we obtain
    a single formula \(\psi(x,y,z)\) and coloring \(S_d^+(G)\) of \(S_d(G)\)
    such that for each \(p \in T_d\) and \(u,v \in V(S_d(G))\), we have
    \(S^+_d(G) \models \psi(p,u,v)\) if and only if 
    \(u,v \in V(G(p))\) and $G^+(p) \models \branch(u,v)$.
    That is, \(\psi(p,u,v)\) holds exactly when in the branching of \(G(p)\),
    \(u\) is the leader of a branch that contains \(v\).

    Consider now the following process and observe that it constructs \(S_{d+1}(G)\):
    \begin{itemize}
        \item Start with the graph \(S_d(G)\) and color it with a bounded number of colors to obtain \(S_d^+(G)\).
        \item For each vertex \(u\) in \(V(G)\), and each color \(c \in C\),
            if there is a (by \Cref{lem:tree_coloring} unique) node \(p \in T_d\) with \(\lambda(p)=c\)
            such that \(u\) is a leader in the branching of \(G(p)\), that is, if \(S_d^+(G) \models \exists z \psi(p,u,z)\),
            then add the copy \((u,c)\) of \(u\) as a child of \(p\) to the tree.

        \item For each vertex \(v\) in \(V(G)\), and each color \(c \in C\),
            if there is a leaf \((u,c)\) with parent \(p\) in the tree such that \(v\) is in the branch of \(p\) with leader \(u\),
            that is if \(S_d^+(G) \models \psi(p,u,v)\),
            then add an edge between \(v\) and \((u,c)\).
    \end{itemize}
    It is easy to see that this process can be executed by a transduction that (1) first colors \(S_d(G)\),
    then (2) creates up to \(|C|\) additional copies of each vertex, then (3) applies a fixed interpretation,
    and finally (4) removes the unused copies.
    This shows that \(S_{d+1}(G)\) can indeed be obtained from \(S_d(G)\) by a fixed transduction.
\end{proof}

\subsection{$\D$ has bounded expansion}
\label{sec:be}

Recall that $\C$ has inherently linear neighborhood complexity.
By \Cref{lem:transducible}, the graph class $\D$ is transducible from $\C$,
and thus $\D$ also has inherently linear neighborhood complexity.
This, together with the fact that $\D$ is weakly sparse (proven below), will show that $\D$ has bounded expansion.

\begin{lemma}
    $\D$ is $\tau$-degenerate for some $\tau$ depending only on $\C$. In particular,  $\D$ is weakly sparse.
\end{lemma}
\begin{proof}
	Consider any $G\in \C$.
	We will prove that in $S(G)$, every vertex belonging to $V(G)$ has degree at most $\tau\coloneqq \sum_{d=0}^h k^d$, where $h$ is the bound on the height of $T(G)$ obtained from \Cref{lem:tree-height} and $k$ is the branching overlap from \Cref{lem:branching_compute}. Note that this implies that $G$ is $\tau$-degenerate as follows. Let $H$ be any subgraph of $S(G)$. We distinguish two cases. First, if $H$ contains some vertex $v$ from $V(G)$, then $H$ contains a vertex of degree at most $\tau$. Second, if $H$ does not contain a vertex from $V(G)$, then $H$ is a subgraph of $T(G)$,  and so $H$ contains a vertex of degree at most $1$. In both cases $H$ contains a vertex of degree at most $\tau$, which means that $S(G)$ is $\tau$-degenerate. Moreover, note that from the $\tau$-degeneracy of $S(G)$ it immediately follows that $S(G)$ cannot contain $K_{\tau+1,\tau+1}$ as a subgraph.	

As $\tau$ depends only on $\C$ and $G$ was chosen arbitrarily from $\C$, this means that $\D$ is $\tau$-degenerate and weakly sparse, as desired.

    To prove the claim about $\tau$ stated above, we will prove by induction on $d$ the following statement: In $S(G)$, every vertex $v \in V(G)$ has at most $k^d$ neighbors among the nodes of $T(G)$ at depth $d$. For $d=0$, $v$ has precisely $k^0=1$ neighbors in $T(G)$, namely the root of $T(G)$. For $d>0$, let $N$ be the set of neighbors of $v$ in $T(G)$ at depth $d-1$. Then by the induction hypothesis we have that $|N| \le k^{d-1}$. From the construction of $T(G)$ we know that all the nodes of $T(G)$ at depth $d$ that are neighbors of $v$ are among the children of nodes in $N$. For any $p \in N$, let $G_1,\ldots, G_m$ be the branching of $G(p)$ used in the construction of $T(G)$. Since the overlap of this branching is at most $k$, we know that $v$ is in at most $k$ children of $p$, and so in $S(G)$, $v$ is adjacent to at most $k$ children of $p$. This means that there are at most $k\cdot|N| \le k\cdot k^{d-1} = k^d$ nodes at depth $d$ in $T(G)$ that are neighbors of $v$ in $S(G)$, as desired.
\end{proof}

\begin{lemma}\label{lem:be}
    $\D$ has bounded expansion. 
\end{lemma}
\begin{proof}
    Assume for contradiction that this is not true. Then, since $\D$ is weakly sparse by the previous lemma, by the characterization of \Cref{lem:BE_subdivisions}, there exists $r \in \N$ such that for every $d \in \N$ there exists $H$ with $\frac{\|H\|}{|H|} > d$ such that some $G\in \D$ contains an $({\le} r)$-subdivision $H'$ of $H$ as an induced subgraph.
    We fix such $r$ for the rest of the proof. 
    We will show that this leads to a contradiction with the assumption that $\D$ has inherently linear neighborhood complexity. 

We first show  that if we consider $G$, $H$ and $H'$ as above, then there are at least $|H|(d-\tau)$ edges of $H$ that are subdivided at least once in $H'$. To see this, first note that $H$ has at least $|H|d$ edges. Next, let $K$ be the subgraph of $H'$ induced by the vertices of $H$. Then $K$ contains precisely the edges of $H$ that are not subdivided in $H'$. Note that $K$ is a subgraph of $G$, and since $G$ is $\tau$-degenerate by the previous lemma, we have that $\frac{\|K\|}{|K|} \le \tau|K| = \tau|H|$. Thus, since there are at most $\tau|H|$ edges of $H$ that are not subdivided in $H'$, the remaining $|H|(d-\tau)$ edges of $H$ are subdivided at least once in $H'$.

We now proceed as follows. Again let $G$, $H$ and $H'$ be as above, and let $H^*$ be the graph obtained from $H$ by removing all edges that are not subdivided at least once in $H'$. Note that $|H| = |H^*|$ and that so by the argument above we have that $\frac{\|H^*\|}{|H^*|} > d- \tau$.
    We now describe a transduction $\T$ that

      transduces a $1$-subdivision of $H^*$ from $G$. This transduction does not use copying and the marking used by $\T$ is described as follows. Let $H'$ denote the induced subgraph of $G$ that is an $({\le} r)$-subdivision of $H$. Mark all principal vertices of $H$ in $H'$ by a unary predicate $P$. Then for every $uv \in E(H^*)$, mark all the internal vertices of the path in $H'$ between $u$ and $v$ that corresponds to the edge $uv$ in $H$ by a unary predicate $Q$, and additionally mark exactly one arbitrary vertex on this path by a unary predicate $R$. The formula used by the transduction simply connects every vertex marked by predicate $R$ to the two vertices marked with predicate $P$ that are reachable from $R$ via a path of length at most $r$ going through vertices marked with $Q$. 
    Then $\T(\D)$ contains $1$-subdivisions \(H''\) of graphs $H^*$ with $\frac{\|H^*\|}{|H^*|} > d- \tau$ for arbitrarily large $d$.
    Let $S$ be the set of principal vertices of $H''$. 
    Since each subdivision vertex of $H''$ corresponds to an edge $e$ of $H^*$ and has its endpoints as neighbors, there are at least $\|H^*\|$ different neighborhoods in $|S|$ in $H''$. Thus for $H'' \in \T(\D)$ we have $| N_{H''}(v)\cap S\ \colon\ v \in V(H'') | > (d-\tau)\cdot |S|$, as desired.
    As \(d\) can be made arbitrarily large, $\T(\D)$ does not have linear neighborhood complexity.
\end{proof}

\subsection{Proof of \Cref{thm:main}}\label{sec:final}

We now conclude the proof of \Cref{thm:main}. By \Cref{sec:bipartite}, it suffices to prove the claim for bipartite inputs.
Given $G\in\C$, the algorithm of \Cref{sec:algo_sparsify} computes $S(G)$ by first building the tree $T(G)$ from successive left/right branchings (\Cref{def:branching,lem:branching_compute}) and then attaching $T(G)$ to~$G$.
The auxiliary results established earlier verify the four requirements of \Cref{thm:main}:
\begin{itemize}[nosep]
    \item \Cref{lem:runtime} bounds the running time by $\Oh(|G|^4)$.

    \item \Cref{lem:recover_G_from_S} provides an interpretation \(I\) such that \(G = I(S(G))\) for every \(G \in \C\).

	\item \Cref{lem:transducible} shows that \(\D \coloneqq  \{S(G) : G \in \C\}\) is transducible from $\C$.

    \item \Cref{lem:be} yields that \(\D\) has bounded expansion.
\end{itemize}

\bibliographystyle{abbrv}
\bibliography{biblio}

\appendix

\section{Sparsification of monadically stable classes}\label{sec:mon-dep}

\subsection{Additional preliminaries}

\paragraph*{Almost nowhere denseness.}
We first recall the notion of almost nowhere dense graph classes. For this, we need some standard definitions about shallow minors.

Recall that a graph $H$ is a {\em{depth-$r$ minor}} of a graph $G$ if $H$ can be obtained from a subgraph of $G$ by contracting mutually disjoint connected subgraphs of radius at most $r$. We define the following parameter:

\begin{align*}\grad_r(G)\coloneqq & \max\left\{\frac{\|H\|}{|H|}\colon \textrm{$H$ is a depth-$r$ minor of $G$}\right\}.
\end{align*}

As proved by Dvo\v{r}\'ak~\cite{dvorak2007asymptotical}, the parameters $\grad_r(\cdot)$ are bounded subpolynomially in nowhere dense classes, in the following sense.

\begin{theorem}[{Dvo\v{r}\'ak~\cite{dvorak2007asymptotical}}]\label{thm:nd-alm-nd}
	Let $\Cc$ be a nowhere dense class of graphs. Then for any graph $G\in \Cc$ and any $r\in \N$ and $\eps>0$, we have
	\[\grad_r(G)\leq \Oh_{r,\eps}(|G|^\eps).\]
\end{theorem}

The notion of almost nowhere denseness takes the conclusion of \cref{thm:nd-alm-nd} as the definition:

\begin{definition}
	A class of graphs $\Cc$ is {\em{almost nowhere dense}} if for every $r\in \N$, $\eps>0$, and a graph $G\in \Cc$, we have
	\[\grad_r(G)\leq \Oh_{r,\eps}(|G|^\eps).\]
\end{definition}

By \cref{thm:nd-alm-nd}, every nowhere dense class is almost nowhere dense. The reverse implication holds for hereditary classes, but not in full generality. It is known (see e.g.~\cite{sparsity,sparsityNotes}) that many graph parameters --- the maximum edge density among depth-$r$ topological minors, weak $r$-coloring numbers, strong $r$-coloring numbers, and $r$-admissibility --- are all tied to the parameters $\grad_r(\cdot)$ by polynomial inequalities. Hence, replacing the parameters $\grad_r(\cdot)$ with any of these families in the definition of almost nowhere denseness yields exactly the same notion. We omit the details, as we will not use these connections here.

\paragraph*{Decompositions.} We now recall the notion of {\em{decompositions}} of graph classes proposed by Braunfeld, Ne\v{s}et\v{r}il, Ossona de Mendez, and Siebertz~\cite{horizons}. The concept is inspired by the classic notion of low treedepth colorings from the field of Sparsity (see e.g.~\cite{sparsity,sparsityNotes}).

\begin{definition}[{\cite[Definition~14]{horizons}}]
	Let $\Pi$ be a property of graph classes. We say that a graph class $\Cc$ admits {\em{quasibounded-size $\Pi$-decompositions}} if there are graph classes $\Dd_1,\Dd_2,\Dd_3,\ldots$, all enjoying $\Pi$, so that the following holds: For every graph $G\in \Cc$, $t\in \N$, and $\eps>0$, there exists a vertex coloring $\lambda\colon V(G)\to C$, for some set $C$ consisting of $\Oh_{t,\eps}(|G|^\eps)$ colors, such that
	\[G[\lambda^{-1}(X)]\in \Dd_t\qquad\textrm{for each $X\subseteq C$ with $|X|\leq t$.}\] 
\end{definition}

We note that the original formulation of Braunfeld et al. speaks about partitions of the vertex set instead of vertex colorings, but this is equivalent to the formulation above.

We will rely on the following observation of Braunfeld et al. that almost nowhere denseness persists under quasibounded-size decompositions.

\begin{lemma}[{\cite[Lemma~34]{horizons}}]\label{lem:decomp-almnd}
	Suppose a graph class $\Cc$ admits quasibounded-size almost nowhere dense decompositions. Then $\Cc$ is itself also almost nowhere dense.
\end{lemma}

\paragraph*{Monadic dependence and nowhere denseness.}
Recall that a graph class $\Cc$ is {\em{monadically dependent}} if the class of all graphs is not transducible from~$\Cc$; and a graph class $\Dd$ is {\em{weakly sparse}} if there is $h\in \N$ such that no member of $\Dd$ contains $K_{h,h}$ as a subgraph. As noted e.g. in~\cite[Corollary 2.3]{nesetril2021linrw_stable}, the following statement is implied by the results of Adler and Adler~\cite{adler2014interpreting} and of Dvo\v{r}\'ak~\cite{Dvorak18}.

\begin{theorem}[follows from~\cite{adler2014interpreting,Dvorak18}]\label{thm:mondep-nd}
	Every monadically dependent graph class that is weakly sparse is actually nowhere dense. 
\end{theorem}

Note that since every monadically stable graph class is also monadically dependent, the conclusion of \cref{thm:mondep-nd} also applies to weakly sparse monadically stable classes: they are, in fact, nowhere dense.

\subsection{Sparsification of monadically stable graph classes to almost nowhere dense classes}

We now argue that the sparsification procedure presented in the main body of the paper, when applied to a graph $G$ from any monadically stable graph class $\Cc$ without the assumption of inherently linear neighborhood complexity, sparsifies $G$ to a graph $H$ that belongs to an almost nowhere dense graph class. Precisely, we prove the following variation on \cref{thm:main}.

\begin{theorem}\label{thm:almost-nowhere-dense-sparsification}
Let $\C$ be a monadically stable graph class.
Then there exists an almost nowhere dense graph class $\D$, an algorithm $\Aa$, and a first-order interpretation $I$ such that the following holds: For any input graph $G \in \C$ on $n$ vertices and $\eps>0$, the algorithm $\Aa$ computes in time $\Oh_\eps(n^{4+\eps})$ a colored graph $H \in \D$ such that $G = I(H)$.
\end{theorem}

As we mentioned, a statement of this kind was proved by Braunfeld et al. as~\cite[Theorem~59]{horizons}. There, the graph $H$ and the interpretation $I$ have a concrete shape given by the notion of a {\em{quasi-bush}}, introduced by Dreier et al.~\cite{bushes_quasibushes}. The goal of this section is to demonstrate that the method proposed in this paper can be used to given an alternative proof of this result. While we do not stick to the formal definition of a quasi-bush as proposed in~\cite{bushes_quasibushes}, the outcome of our sparsification procedure is, in fact, very similar in spirit.

\bigskip

The remainder of this section is devoted to the proof of \cref{thm:almost-nowhere-dense-sparsification}.
We assume the reader's familiarity with the proof of \cref{thm:main}, presented in \cref{sec:mainlemma,sec:sparsify}. Here, we only sketch how the reasoning can be amended to obtain \cref{thm:almost-nowhere-dense-sparsification}.

To replace the assumption about linear neighborhood complexity, 	
we will use the result of Dreier, Eleftheriadis, M\"ahlmann, McCarty, Pilipczuk, and Toru\'nczyk~\cite{stable_MC} that monadically stable classes have almost linear neighborhood complexity.
	
	\begin{theorem}[\cite{stable_MC}]\label{thm:neicomp-stable}
		Let $\C$ be a fixed monadically stable graph class and let $\eps>0$. Then for every graph $G\in \C$ and a nonempty vertex subset $S\subseteq V(G)$, we have
		\[|\{N(v)\cap S\colon v\in V(G)\}|\leq \Oh_{\eps}(|S|^{1+\eps}).\]
	\end{theorem}

Note that since any class transducible from a monadically stable class is again monadically stable, the bound postulated in \cref{thm:neicomp-stable}  holds also for every graph class transducible from $\Cc$.

\paragraph{The analogue of \cref{lem:main}.}
	The first step is to obtain an analogue of
	the main technical lemma, \cref{lem:main}. Precisely, we argue the following statement.
	
	\begin{lemma}
		\label{lem:main-analogue}
		Fix a monadically stable class of bipartite graphs $\C$. Then there is an algorithm with running time $\mathcal{O}(|G|^4)$ that for any given  graph $G = (A,B,E)$ from $\C$ without isolated vertices, computes a partition $\F= \{ P_1,\ldots, P_{|\F|}\}$ of $B$ and vertices $\ell_1,\ldots, \ell_{|\F|}$ belonging to $A$ such that:
		\begin{itemize}[itemsep=1pt]
			\item $P_i \subseteq N_G(\ell_i)$, for each $i\in \{1,\ldots,|\F|\}$; and
			\item every $v \in A$ has a neighbor in at most $k^*$ parts of $\F$, where $k^\star$ is an integer satisfying $k^\star\leq \Oh_\eps(|G|^\eps)$, for any $\eps>0$.
			That is, we have
			\(|\{P \in \F\ |\ N_G(v) \cap P \not= \emptyset\}| \le k^*.\)
		\end{itemize}
		Moreover, the partition $\F$ of $B$ is almost definable in the following sense: There is a coloring $\mu\colon \F\to D$ for some set of colors $D$ satisfying $|D|\leq \Oh_{\eps}(|G|^\eps)$ for any $\eps>0$, so that for each color $c\in D$, there is a unary coloring $G^+_c$ of $G$ and a formula $\leader_c(x,y)$, of length bounded by a universal constant (i.e., independent of $G$ or $\C$), such that for any $u\in A$ and $v\in B$, we have \[G^+_c\models \leader_c(u,v)\quad \textrm{if and only if}\quad \textrm{there exists } i\in \{1,\ldots,|\F|\}\textrm{ such that } \mu(P_i)=c\textrm{, }u=\ell_i,\textrm{ and }v\in P_i.\]
	\end{lemma}

	In summary, the differences between \cref{lem:main} and \cref{lem:main-analogue} are the following:
	\begin{itemize}[nosep]
		\item In \cref{lem:main}, $k^*$ is a constant depending only on $\Cc$. In \cref{lem:main-analogue}, we allow this parameter to be superconstant, but bounded subpolynomially in the size of the considered graph $G$.
		\item The definability property asserted in \cref{lem:main} is replaced by the almost definability property. In essence, instead of requiring the leadership relation to be fully definable using a single formula, we require it to be definable ``piecewise'': $\F$ can be broken into $\Oh_{\eps}(|G|^\eps)$ colors so that the leadership property can be defined within each color separately.
	\end{itemize}

	Let us now explain the amendments that need to be applied to the proof of \cref{lem:main} in order to obtain the proof of \cref{lem:main-analogue}.
	Recall that the proof of \cref{lem:main} consists of: the algorithm to construct the partition $\F$ (described in \cref{sec:algo-main-lemma}); the proof of \cref{lem:bd_degree} (presented in \cref{sec:lem:bd_degree}), which provides a bound on $k^*$; and the proof of the definability property (presented in \cref{sec:definability}).
	
	As for the algorithm, we apply exactly the same algorithm to compute $\F$ as the one described in \cref{sec:algo-main-lemma}. (Note that the algorithm itself depends only on the input graph $G$ and is oblivious of the graph class $\C$ from which $G$ is drawn.) Consequently, the computed family $\F=\{P_1,\ldots,P_{|\F|}\}$ is again a partition of $B$, and each part $P_i$ is supplied with a leader $\ell_i\in A$ satisfying $P_i\subseteq N_G(\ell_i)$.
	
	As for \cref{lem:bd_degree}, we amend it so that it provides a bound $k^*\leq \Oh_{\eps}(|G|^\eps)$, for any $\eps>0$. Recall that the proof of \cref{lem:bd_degree} relies on two lemmas: \cref{lem:near_twins,lem:quotient_no_halfgraph}. \cref{lem:quotient_no_halfgraph} asserts that the graph $G(A,\F)$, encoding the adjacency between the vertices of $A$ and the parts of $\F$, excludes a half-graph of order $t$ as an induced subgraph, for some constant $t$ depending on $\C$. The proof of this lemma only uses the assumption that $\C$ is monadically stable, so this lemma holds in our setting as well without any changes; and $t$ is a constant depending only on $\C$. In turn, \cref{lem:near_twins} asserts that if $a_1,\ldots,a_n$ is the ordering of $A$ computed by the algorithm, then for each $i<n$ there exists $j>i$ such that $a_i$ and $a_j$ are $k^\circ$-near twins in $G(A,\F)$, where $k^\circ$ is the constant provided by 
	\cref{lem:bip-graph-near-twins}. Recall that \cref{lem:bip-graph-near-twins} is a reformulation of the Haussler's Packing Lemma for set systems with linear neighborhood complexity, and this reformulation provides a pair of $k^\circ$-near twins, where $k^\circ$ is a constant. For set systems with almost linear neighborhood complexity, in the sense of \cref{thm:neicomp-stable}, the same reformulation provides a pair of $k^\circ$-near twins, where $k^\circ\leq \Oh_\eps(|G|^\eps)$. So if we use this statement instead, we obtain a variant of \cref{lem:near_twins} that assumes only monadic stability of $\C$ and provides a bound $k^\circ\leq \Oh_\eps(|G|^\eps)$. Finally,  \cref{lem:near_twins,lem:quotient_no_halfgraph} are combined in the proof of \cref{lem:bd_degree} using \cref{lem:Rose_lemma}, providing a bound $k^*\leq h(k,t)+1$, where $h$ is the function given by \cref{lem:Rose_lemma}. Inspection of the proof of \cref{lem:Rose_lemma} in \cite{RoseLemma} yields that we in fact have $h(k,t)\leq f(t)\cdot (k+1)$ for some function $f\colon \N\to \N$. Combining this with $k^\circ\leq \Oh_\eps(|G|^\eps)$, we obtain the desired bound $k^*\leq \Oh_\eps(|G|^\eps)$.
	
	Finally, we need to argue the almost definability property. For this, we set $\mu(P_i)=\lambda(i)$, where $\lambda$ is the coloring of $[|\F|]$ with $2k^*+1=\Oh_\eps(|G|^\eps)$ colors constructed in \cref{sec:definability}. For each color $c$, we set $G_c^+$ to be the unary extension of $G$ that marks all parts of color $c$ and their leader with a single additional unary predicate $M_c$. Then, similarly as in \cref{sec:definability}, we can simply take
	\[\leader_c(x,y)\coloneqq L(x)\wedge M_c(x)\wedge M_c(y)\wedge \adj(x,y).\]
	 
	This concludes the sketch of the proof of \cref{lem:main-analogue}.

	\paragraph*{The full construction.}
	Armed with \cref{lem:main-analogue}, which is the analogue of \cref{lem:main}, we proceed to amending the reasoning of \cref{sec:sparsify}. We fix $\eps>0$ to be used throughout the proof.
	The reduction to the setting of bipartite graphs, presented in \cref{sec:bipartite}, can be applied in the same way.
	
	We use the same definition of branchings (\cref{def:branching}). The proof of \cref{lem:branching_compute} can be adjusted to prove the following analogue:
	\begin{lemma}\label{lem:branching_compute-analogue}
		There exists an algorithm with runtime $\Oh(|G|^4)$ that computes for any $G \in \C$ a left (or right)
		branching of overlap at most $k$, where $k\leq \Oh_\eps(|G|^\eps)$.
		
		Furthermore, the computed branching $\cal B$ is almost definable in the following sense: There is a coloring $\mu\colon {\cal B}\to D$ for a set $D$ consisting of $\Oh_\eps(|G|^\eps)$ colors so that for each $c\in D$, there is a unary coloring $G_c^+$ of $G$ and a first-order formula $\branch_c(x,y)$, of length bounded by a universal constant (independent of $G$, $\C$, or $\eps$), that satisfy the following: $G_c^+\models \branch_c(u,v)$ if and only if there exists $H\in \cal B$ with $\mu(H)=c$ such that $v\in V(H)$ and $u$ is the leader of $H$.
	\end{lemma}

	Recall that \cref{lem:branching_compute} is a simple corollary of \cref{lem:main}, where the partition $\F$ computed by the algorithm of \cref{lem:main} is turned into a branching by expanding each part of $\F$ to its neighborhood. Applying the same argument with \cref{lem:main} replaced with \cref{lem:main-analogue} instead, yields a proof of \cref{lem:branching_compute-analogue}: the $\Oh_\eps(|G|^\eps)$ bound on the overlap of the computed branching is directly implied by the  $k^*\leq \Oh_\eps(|G|^\eps)$ bound of \cref{lem:main-analogue}, and the almost definability property of the branching, asserted in \cref{lem:branching_compute-analogue}, is a direct consequence of the almost definability property asserted in \cref{lem:main-analogue}.

	Next, the algorithm to compute the tree $T(G)$ and the graph $S(G)$, presented in \cref{sec:algo_sparsify}, remains unchanged, except for the usage of \cref{lem:branching_compute-analogue} instead of \cref{lem:branching_compute}. Note that the definitions of $T(G)$ and of $S(G)$ do not depend on $\eps$. Consequently, we define the class $\D\coloneqq \{S(G)\colon G\in \C\}$  in exactly the same way.
	The proof that $T(G)$ has height bounded by a constant depending only on $\C$, presented in \cref{sec:height}, applies verbatim, as it relies only on the exclusion of half-graphs, following from the assumption of monadic stability. We denote this constant by $h$. 
	
	As for the running time of the algorithm to compute $T(G)$ and $S(G)$, the analysis from the proof of \cref{lem:runtime} provides a bound of $\Oh(k^h\cdot |G|^4)$, where $k$ is the bound on the overlap of branchings provided by \cref{lem:branching_compute}. With $k$ being bounded by $\Oh_\eps(|G|^\eps)$ (as asserted by \cref{lem:branching_compute-analogue}) and $h$ being a constant depending only on $\C$, we have that $k^h\leq \Oh_\eps(|G|^{h\eps})$, which becomes $\Oh_\eps(|G|^\eps)$ after rescaling $\eps$ by a factor of $h$. Therefore, the running time of the algorithm becomes $\Oh_\eps(|G|^{4+\eps})$.
	
	The interpretation $I$ satisfying $G=I(S(G))$, proposed in \cref{sec:interp}, may remain unchanged, as its construction relies only on $h$.
	
	We are left with the crux of the proof: arguing that the constructed class $\Dd$ is almost nowhere dense. By \cref{lem:decomp-almnd}, it suffices to prove that $\Dd$ admits quasibounded-size almost nowhere dense decompositions. In fact, we shall prove that $\Dd$ admits quasibounded-size nowhere dense decompositions, which is sufficient by \cref{thm:nd-alm-nd}.
	
	To this end, fix a graph $G\in \Cc$ and consider the graph $S(G)\in \Dd$. Our first goal is to define a suitable vertex coloring of $S(G)$ with $\Oh_\eps(|G|^\eps)$ colors. Recall that the vertex set of $S(G)$ is the disjoint union of the vertex set of $G$ and the node set of $T(G)$. Consider any non-leaf node $p$ of~$T(G)$. Recall that in the construction we have computed a branching of $G(p)$ (with the isolated vertices removed); call it ${\cal B}(p)$. For each graph $G_i\in {\cal B}(p)$ we have constructed a child $p_i$ of $p$ associated with the graph $G(p_i)$, defined as the bipartite complement of $G_i$. Further, we have a coloring $\mu_p\colon {\cal B}(p)\to D_p$, for some set $D_p$ of $\Oh_{\eps}(|G|^\eps)$ colors, that witnesses the almost definability property of the branching ${\cal B}(p)$ (as defined in \cref{lem:branching_compute-analogue}). We may assume that $D_p=D$ for some set $D$ of $\Oh_{\eps}(|G|^\eps)$ colors, fixed for the whole tree $T(G)$. We now define a coloring $\mu\colon V(T(G))\to D$ by setting $\mu(q)\coloneqq \mu_p(G_i)$ for every non-root node $q$, where $p$ is the parent of $q$ and $G_i\in {\cal B}(p)$ is such that $G(q)$ is the bipartite complement of $G_i$. The color of the root of $T(G)$ is set to be an arbitrary element of $D$.
	
	On the other hand, keeping in mind that the overlap of the branching ${\cal B}(p)$ is bounded by $k=\Oh_\eps(|G|^\eps)$, we may repeat the reasoning of \cref{lem:branching_graph_sparse,lem:tree_coloring} to find another coloring $\nu\colon V(T(G))\to D'$, again for a set of colors $D'$ of size $\Oh_{\eps}(|G|^\eps)$, satisfying the following property: for all $p,q\in V(T(G))$ with $\nu(p)=\nu(q)$, we have $V(G(p))\cap V(G(q))=\emptyset$. We define $\sigma\colon V(T(G))\to D\times D'$ to be the product coloring of $\mu$ and $\nu$:  $\sigma(p)=(\mu(p),\nu(p))$ for every node $p$. Then $\sigma$ satisfies both the discussed properties:
	\begin{itemize}[nosep]
		\item For each non-leaf node $p$ of $T(G)$, the coloring $\sigma$ restricted to the children of $p$ witnesses the almost definability of ${\cal B}(p)$.
		\item For all $p,q\in V(T(G))$ with $\sigma(p)=\sigma(q)$, we have $V(G(p))\cap V(G(q))=\emptyset$.
	\end{itemize}
	We remark that the first assertion above requires an easy adjustment of the unary expansions $G^+_c$ and the formulas $\branch_c(x,y)$.
	
	Finally, for every node $p$ of $T$, we define
	\[\lambda(p)\coloneqq \sigma(p_1)\sigma(p_2)\ldots \sigma(p_\ell)\in (D\times D')^{\leq h},\]
	where $p_1,p_2\ldots,p_\ell$ are the nodes on the root-to-$p$ path in $T(G)$, with $p_1$ being the root and $p_\ell=p$, and $C\coloneqq (D\times D')^{\leq h}$ is the set of words of length at most $h$ over the alphabet $D\times D'$. By rescaling $\eps$ by a factor of $2h$ (which is a constant depending on $\Cc$ only), we have $|C|\leq \Oh_\eps(|G|^\eps)$. We extend the coloring $\lambda$ also to the vertices of $G$ by mapping all of them to the same fresh color; this increases $|C|$ by $1$. Thus, $\lambda$ is a vertex coloring of $S(G)$.
	
	Having defined the coloring $\lambda$ for each graph $S(G)\in \Dd$, we verify that these colorings witness that $\Dd$ admits quasibounded-size nowhere dense decompositions. For this, we prove the following claim.
	
	\begin{claim}\label{cl:heart}
		Fix $t\in \N$.
		For every $G\in \Cc$, if we consider the coloring $\lambda\colon V(S(G))\to C$ constructed above, then for every subset of colors $X\subseteq C$ with $|X|\leq t$ we have the following:
		\begin{itemize}[nosep]
			\item $S(G)[\lambda^{-1}(X)]$ is $t$-degenerate, and
			\item $S(G)[\lambda^{-1}(X)]\in \ST^t(G)$, where $\ST^t$ is a transduction depending only on $\Cc$ and $t$.
		\end{itemize}   
	\end{claim}
	\begin{proof}
		To see that $S(G)[\lambda^{-1}(X)]$ is $t$-degenerate, note that every vertex $u\in \lambda^{-1}(X)\cap V(G)$ has degree at most $t$ in $S(G)[\lambda^{-1}(X)]$, for all its neighbors in $V(T(G))\cap \lambda^{-1}(X)$ have pairwise different colors under $\sigma$, hence also under $\lambda$. Therefore, every subgraph of $S(G)[\lambda^{-1}(X)]$ either contains a vertex of $\lambda^{-1}(X)\cap V(G)$, which has degree at most $t$, or is a subtree of $T(G)[\lambda^{-1}(X)]$, and hence contains a vertex of degree at most $1$.
		
		For the second assertion, let $X'$ be the closure of $X$ under taking nonempty prefixes. That is, for each color $c\in X$ (which, recall, is a word over $D\times D'$ of length at most $h$), add all the nonempty prefixes of $c$ to~$X'$. Thus, $X'\supseteq X$ and $|X'|\leq ht$. As $S(G)[\lambda^{-1}(X)]$ is an induced subgraph of $S(G)[\lambda^{-1}(X')]$, it suffices to prove that $S(G)[\lambda^{-1}(X')]\in \ST^t(G)$, for some fixed transduction $\ST^t$ depending only on $\Cc$ and $t$. To argue this, we may repeat the reasoning from the proof of \cref{lem:transducible}, but with the following amendment: in the definition of $S_d(G)$, we additionally remove all the vertices that do not belong to $\lambda^{-1}(X')$. (We may assume here without loss of generality that $V(G)\subseteq \lambda^{-1}(X')$, thus this only removes some the nodes of $V(T(G))\setminus \lambda^{-1}(X')$.) Since $X'$ is closed under taking nonempty prefixes by definition, we see that $V(S_d(G))\cap V(T(G))$ is a prefix $T(G)$, that is, a set of nodes of $T(G)$ closed under taking ancestors. With this observation, the transduction $\ST^t_d$, used to transduce $S_d(G)$ from $S_{d-1}(G)$, can be constructed in a similar way, except for the following adjustment.
		To define the formula $\psi(x,y,z)$, instead of using, for each node $p$ of $T(G)$, the unary coloring $G^+(p)$, we consider the unary colorings $G^+_c(p)$ for all the colors $c\in X''$, where $X''\subseteq D\times D'$ is the set of all the colors featured in the elements of $X$ (thus $|X''|\leq ht$). Also, instead of considering all the colors $c\in C$ that can be possibly used by $\lambda$, we restrict attention only to the colors $c\in X''$.
		Then, from the formulas $\branch_c(x,y)$ for $c\in X''$ we may piece together
		a formula $\psi(x,y,z)$ that works as in the proof of \cref{lem:transducible}, but is applied to a unary coloring of $G(p)$ with a bounded number of unary predicates, and expresses the condition that $y$ is the leader of the branch containing $z$ only under the assumption that this branch is of color belonging to $X''$. With such a formula $\psi$ in place, in the same manner as in the proof of \cref{lem:transducible} we may obtain a transduction $\ST^t_d$ that constructs, for each $p\in T_d\cap \lambda^{-1}(X')$, the children of $p$ of colors belonging to $X'$, together with all the edges connecting them to the vertices of $G$. 
		
		As in the proof of \cref{lem:transducible}, the final transduction~$\ST^t$ is obtained by composing all the transductions~$\ST^t_d$, for $d$ ranging from $1$ to $h$, and a transduction that clears all the edges of $G$.
	\renewcommand{\qed}{\hfill$\lrcorner$}\end{proof}
		
	Given \cref{cl:heart}, we define $\Ee_t$ to be the class of all $t$-degenerate graphs belonging to $\ST^t(\Cc)$. Thus, $\Ee_t$ is in particular weakly sparse. Since $\Ee_t$ is transducible from a monadically dependent (even monadically stable) class $\Cc$, it follows that $\Ee_t$ is also monadically dependent. Hence, from \cref{thm:mondep-nd} we conclude that $\Ee_t$ is actually nowhere dense. Therefore, by \cref{cl:heart}, the classes $\Ee_1,\Ee_2,\Ee_3,\ldots$ and the colorings $\lambda$ witness that $\Dd$ admits quasibounded-size nowhere dense decompositions, as claimed. This finishes the proof of \cref{thm:almost-nowhere-dense-sparsification}.

\end{document}